\documentclass[11pt]{article}

\usepackage[utf8]{inputenc} 
\usepackage[T1]{fontenc}    
\usepackage{hyperref}       
\usepackage{url}            
\usepackage{booktabs}       
\usepackage{amsfonts}       
\usepackage{nicefrac}       
\usepackage{microtype}      

\usepackage{microtype}
\usepackage{graphicx}
\graphicspath{{Figures/}}
\usepackage{booktabs} 
\usepackage{caption}
\usepackage{subcaption}
\usepackage{amsmath}
\usepackage{amsfonts}
\usepackage{amssymb}
\usepackage{authblk}
\usepackage{cite}
\usepackage[linesnumbered,ruled,vlined ]{algorithm2e}

\usepackage{amsopn,amssymb,amsthm,amsmath,bm}
\usepackage{paralist}
\usepackage{multirow}
\usepackage{multicol}
\usepackage{float}
\usepackage{thmtools}
\usepackage{thm-restate}
\usepackage{framed}
\usepackage{hyperref}
\usepackage{enumitem}

\newcommand{\A}{\mathbf{A}} 
\newcommand{\ba}{\mathbf{a}} 
\newcommand{\tA}{\tilde{\mathbf{A}}} 
\newcommand{\x}{\mathbf{x}} 
\newcommand{\bb}{\mathbf{b}}

\newcommand{\y}{\mathbf{y}}
\newcommand{\p}{\mathbf{p}}
\newcommand{\w}{\mathbf{w}}

\newcommand{\s}{\mathbf{S}}
\newcommand{\z}{\mathbf{z}}
\newcommand{\I}{\mathbf{I}}
\newcommand{\X}{\mathbf{X}}
\newcommand{\R}{\mathbb{R}}
\newcommand{\h}{\mathbf{H}}
\newcommand{\hh}{\hat{\mathbf{H}}}
\newcommand{\ddh}{\hat{\mathbf{H}}^{\dagger}}
\newcommand{\hd}{\mathbf{H}^{\dagger}}
\newcommand{\Del}{\mathbf{\Delta}}
\newcommand{\nn}{\nonumber}

\newcommand{\U}{\mathbf{U}}
\newcommand{\V}{\mathbf{V}}
\newcommand{\hU}{\hat{\mathbf{U}}} 

\newcommand{\g}{\mathbf{g}}
\newcommand{\W}{\mathbf{W}}
\newcommand{\Z}{\mathbf{Z}}

\usepackage{xcolor}

\usepackage{amsthm}
\newtheorem{theorem}{Theorem}[section]

\newtheorem{lemma}[theorem]{Lemma}

\newtheorem{remark}{Remark}

\title{OverSketched Newton: Fast Convex Optimization for Serverless Systems}

\author[1]{Vipul Gupta}
\author[1]{Swanand Kadhe}
\author[1]{Thomas Courtade}
\author[2]{Michael W. Mahoney}
\author[1]{Kannan Ramchandran}

\affil[1]{Department of EECS, University of California, Berkeley}
\affil[2]{ICSI and Statistics Department, University of California, Berkeley}
\date{\vspace{-5ex}}

\begin{document}

\maketitle

\begin{abstract}
Motivated by recent developments in serverless systems for large-scale computation as well as improvements in scalable randomized matrix algorithms, we develop OverSketched Newton, a randomized Hessian-based optimization algorithm to solve large-scale convex optimization problems in serverless systems. 
OverSketched Newton leverages matrix sketching ideas from Randomized Numerical Linear Algebra to compute the Hessian approximately. 
These sketching methods lead to inbuilt resiliency against stragglers that are a characteristic of serverless architectures.
Depending on whether the problem is strongly convex or not, we propose different iteration updates using the approximate Hessian.
For both cases, we establish convergence guarantees for OverSketched Newton and  empirically validate our results by solving large-scale supervised learning problems on real-world datasets. 
Experiments demonstrate a reduction of ${\sim}50\%$ in total running time on AWS Lambda, compared to state-of-the-art distributed optimization schemes.
\end{abstract}


\section{Introduction}
\label{intro}


In recent years, there has been tremendous growth in users performing distributed computing operations on the cloud, largely due to extensive and inexpensive commercial offerings like Amazon Web Services (AWS), Google Cloud, Microsoft Azure, etc. Serverless platforms---such as AWS Lambda, Cloud functions and Azure Functions---penetrate a large user base by provisioning and managing the servers on which the computation is performed.
These platforms abstract away the need for maintaining servers, since this is done by the cloud provider and is hidden from the user---hence the name \emph{serverless}. 
Moreover, allocation of these servers is done expeditiously which provides greater elasticity and easy scalability. For example, up to ten thousand machines can be allocated on AWS Lambda in less than ten seconds \cite{serverless_computing, serverless2, pywren, numpywren}. 

The use of serverless systems is gaining significant research traction, primarily due to its massive scalability and convenience in operation. It is forecasted that the market share of serverless will grow by USD 9.16 billion during 2019-2023 (at a CAGR of 11\%) \cite{serverless_technavio}.
Indeed, according to the \textit{Berkeley view on Serverless Computing} \cite{berkeley_view},  serverless systems are expected to dominate the cloud scenario and become the default computing paradigm in the coming years
while client-server based computing will witness a considerable decline.
For these reasons, using serverless systems for large-scale computation has garnered significant attention from the systems community \cite{pywren,numpywren, deep_serverless,ibm_serverless, kth_serverless,serverless_ml,dnn_training_serverless,cirrus}.

\begin{figure}[t!]
        \centering
        \includegraphics[scale=0.375]{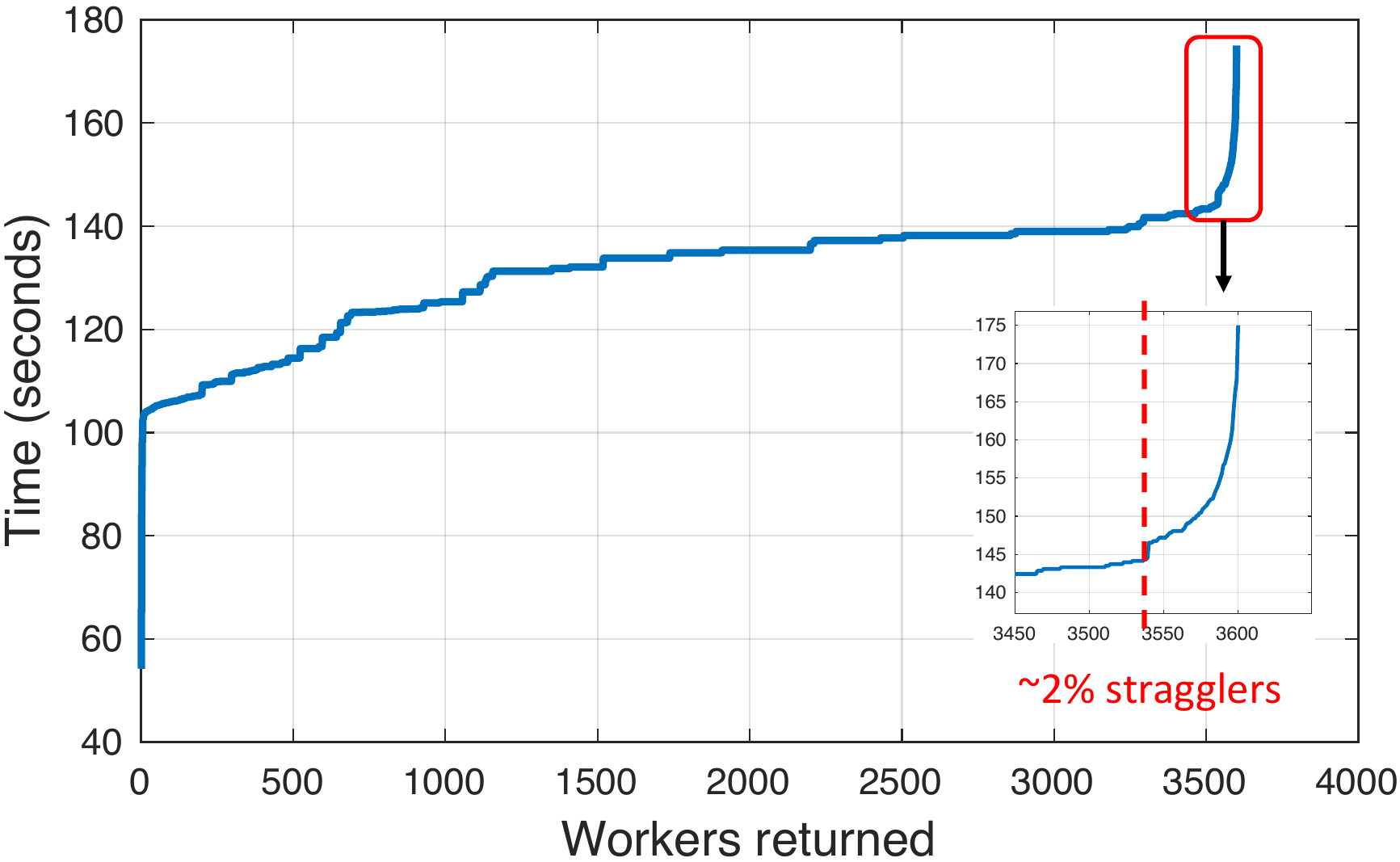}
        \caption{\small Average job times for 3600 AWS Lambda nodes over 10 trials for distributed matrix multiplication. The median job time is around $135$ seconds, and around $2\%$ of the nodes take up to $180$ seconds on average.}
    \label{fig:stragglers}
\end{figure}

Due to several crucial differences between the traditional High Performance Computing (HPC) / serverful and serverless architectures, existing distributed algorithms cannot, in general, be extended to serverless computing. 
First, unlike \emph{serverful} computing, the number of inexpensive workers in serverless platforms is flexible, often scaling into the thousands \cite{pywren,numpywren}. 
This heavy gain in the computation power, however, comes with the disadvantage that 
the commodity workers in serverless architecture are ephemeral and have low memory.%
\footnote{For example, serverless nodes in AWS Lambda, Google Cloud Functions and Microsoft Azure Functions  have a maximum memory of 3 GB, 2 GB and 1.5 GB, respectively, and a maximum runtime of 900 seconds, 540 seconds and 300 seconds, respectively (these numbers may change over time). }
The ephemeral nature of the workers in serverless systems requires that new workers should be invoked every few iterations and data should be communicated to them.
Moreover, the workers do not communicate amongst themselves, and instead they read/write data directly from/to a single high-latency data storage entity (e.g., cloud storage like AWS S3 \cite{pywren}). 

Second, unlike HPC/serverful systems, nodes in the serverless systems suffer degradation due to what is known as {\it system noise}.
This can be a result of limited availability of shared resources, hardware failure, network latency, etc.~\cite{tailatscale, hoefler}. 
This results in job time variability, and hence a subset of much slower nodes, often called \emph{stragglers}.
These stragglers significantly slow the overall computation time, especially in large or iterative jobs. 
In Fig. 1, we plot the running times for a distributed matrix multiplication job with 3600 workers on AWS Lambda and demonstrate the effect of stragglers on the total job time.
In fact, our experiments consistently demonstrate that at least $2\%$ workers take significantly longer than the median job time, severely degrading the overall efficiency of the system.

Due to these 
issues, first-order methods, e.g., gradient descent and Nesterov Accelerated Gradient (NAG) methods, tend to perform poorly on distributed serverless architectures \cite{hellerstein}. 
Their slower convergence is made worse on serverless platforms due to persistent stragglers. 
The straggler effect incurs heavy slowdown due to the accumulation of tail times as a result of a subset of slow workers occurring in each iteration. 

Compared to first-order optimization algorithms, second-order methods---which use the gradient as well as Hessian information---enjoy superior convergence rates.  
For instance, Newton's method enjoys quadratic convergence for strongly convex and smooth problems, compared to the linear convergence of gradient descent \cite{nocedal-wright}.
Moreover, second-order methods do not require step-size tuning and unit step-size provably works for most problems.  
These methods have a long history in optimization and scientific computing (see, e.g., \cite{nocedal-wright}), but they are less common in machine learning and data science.
This is partly since stochastic first order methods suffice for downstream problems \cite{bottou-nocedal-optimization-ml} and partly since naive implementations of second order methods can perform poorly \cite{whitening-second-order}.
However, recent theoretical work has addressed many of these issues \cite{fred1,fred2,Fred:non-convex,newton_mr,mert}, and recent implementations have shown that high-quality implementations of second order stochastic optimization algorithms can beat state-of-the-art in machine learning applications~\cite{giant,FangGPU,pyhessian,adahessian,second-order-made-practical} in traditional systems.


\subsection{Main Contributions}

In this paper, we argue that second-order methods are highly compatible with serverless systems that provide extensive computing power by invoking thousands of workers but are limited by the communication costs and hence the number of iterations; and, to address the challenges of ephemeral workers and stragglers in serverless systems, we propose and analyze a randomized and distributed second-order optimization algorithm, called {\it OverSketched Newton}.
OverSketched Newton uses the technique of matrix sketching from Sub-Sampled Newton (SSN) methods~\cite{fred1,fred2,Fred:non-convex,newton_mr}, which are based on sketching methods from Randomized Numerical Linear Algebra (RandNLA)~\cite{woodruff_now,Mah-mat-rev_BOOK,gittens-big-data}, to obtain a good approximation for the Hessian, instead of calculating the full Hessian. 

OverSketched Newton has two key components. 
For straggler-resilient Hessian calculation in serverless systems, we use the sparse sketching based randomized matrix multiplication method from~\cite{oversketch}. 
For straggler mitigation during gradient calculation, we use the recently proposed technique based on error-correcting codes to create redundant computation \cite{kangwook1,tavor,local_codes}.
We prove that, for strongly convex functions, the local convergence rate of OverSketched Newton is linear-quadratic, while its global convergence rate is linear. 
Then, going beyond the usual strong convexity assumption for second-order methods, we adapt OverSketched Newton using ideas from \cite{newton_mr}. 
For such functions, we prove that a linear convergence rate can be guaranteed with OverSketched~Newton. To the best of our knowledge, this is the first work to prove convergence guarantees for weakly-convex problems when the Hessian is computed approximately using ideas from RandNLA.


We extensively evaluate OverSketched Newton on AWS Lambda using several real-world datasets obtained from the LIBSVM repository \cite{libsvm}, and we compare OverSketched Newton with several first-order (gradient descent, Nesterov's method, etc.) and second-order (exact Newton's method \cite{nocedal-wright}, GIANT \cite{giant}, etc.) baselines for distributed optimization. 
We further evaluate and compare different techniques for straggler mitigation, such as speculative execution, coded computing \cite{kangwook1,tavor}, randomization-based sketching \cite{oversketch} and gradient coding \cite{grad_coding}. 
We demonstrate that
OverSketched Newton is at least 9x and 2x faster than state-of-the-art first-order and second-order schemes, respectively, in terms of end-to-end training time on AWS Lambda. 
Moreover, we show that OverSketched Newton on serverless systems outperforms existing distributed optimization algorithms in serverful systems by at least $30\%$.

\subsection{Related Work}

Our results tie together three quite different lines of work, each of which we review here briefly.

{\bf Existing Straggler Mitigation Schemes:}
Strategies like speculative execution have been traditionally used to mitigate stragglers in popular distributed computing frameworks like Hadoop MapReduce \cite{mapreduce} and Apache Spark \cite{spark}. Speculative execution works by detecting workers that are running slower than expected and then allocating their tasks to new workers without shutting down the original straggling task. The worker that finishes first communicates its results. 
This has several drawbacks, e.g. constant monitoring of tasks is required and late stragglers can still hurt the efficiency.


Recently, many coding-theoretic ideas have been proposed to introduce redundancy into the distributed computation for straggler mitigation
(e.g. see \cite{kangwook1, tavor, grad_coding, local_codes,poly_codes,grover_inverse}). 
The idea of coded computation is to generate redundant copies of the result of distributed computation by encoding the input data using error-correcting-codes. These redundant copies are then used to decode the output of the missing stragglers. Our algorithm to compute gradients in a distributed straggler-resilient manner uses codes to mitigate stragglers, and we compare our performance with speculative~execution.




{\bf Approximate Second-order Methods:} 
In many machine learning applications, where the data itself is noisy, using the exact Hessian is not necessary. 
Indeed, 
using ideas from 
RandNLA, 
one can prove convergence guarantees for SSN methods on a single machine, when the Hessian is computed approximately \cite{mert, fred1, fred2,Fred:non-convex}. 
To accomplish this, many sketching schemes can be used (sub-Gaussian, Hadamard, random row sampling, sparse Johnson-Lindenstrauss, etc. \cite{woodruff_now, Mah-mat-rev_BOOK}), but these methods cannot tolerate stragglers, and thus they do not perform well in serverless environments.

This motivates the use of the \emph{OverSketch} sketch from our recent work in~\cite{oversketch}. 
OverSketch has many nice properties, like subspace embedding, sparsity, input obliviousness, and amenability to distributed implementation. 
To the best of our knowledge, this is the first work to prove and evaluate convergence guarantees for algorithms based on OverSketch. 
Our guarantees take into account the  amount of communication at each worker and the number of stragglers, both of which are a property of distributed systems. 


There has also been a growing research interest in designing and analyzing distributed implementations of stochastic second-order methods~ \cite{dane,disco,aide,giant, jaggi18,cocoa}. 
However, these implementations are tailored for serverful distributed systems. 
Our focus, on the other hand, is on serverless systems. 

{\bf Distributed Optimization on Serverless Systems:} 
Optimization over the serverless framework has garnered significant interest from the research community. However, these works either evaluate and benchmark existing algorithms (e.g., see \cite{kth_serverless,serverless_ml,dnn_training_serverless}) or focus on designing new systems frameworks for faster optimization (e.g., see \cite{cirrus}) on serverless.
To the best of our knowledge, this is the first work that proposes a large-scale distributed optimization algorithm that specifically caters to \emph{serverless architectures} with \emph{provable convergence guarantees}. We exploit the advantages offered by serverless systems while mitigating the drawbacks such as stragglers and additional overhead per invocation of workers.

\section{Newton's Method: An Overview}

We are interested in solving on serverless systems in a distributed and straggler-resilient manner problems of the~form:
\begin{equation}\label{min_f}
f(\w^*) = \min_{\w \in \R^d} f(\w),
\end{equation}
where $f: \mathbb{R}^d \rightarrow \mathbb{R} $ is a closed and convex function bounded from below. 
In the Newton's method, the update at the $(t+1)$-th iteration is obtained by minimizing the Taylor's expansion of the objective function $f(\cdot)$ at $\w_t$, that is
\begin{align}\label{newton_update}
\w_{t+1} = \arg\min_{\w\in \R^d} \Big\{f(\w_t) + \nabla f(\w_t)^T(\w - \w_t) 
\nn \\ 
+ \frac{1}{2}(\w - \w_t)^T\nabla^2f(\w_t)(\w - \w_t)\Big\}.
\end{align}
For strongly convex $f(\cdot)$, that is, when $\nabla^2f(\cdot)$ is invertible, Eq. \eqref{newton_update} becomes
$\w_{t+1} = \w_t - \h_t^{-1}\nabla f(\w_t),$
where $\h_t = \nabla^2f(\w_t)$ is the Hessian matrix at the $t$-th iteration. 
Given a good initialization and assuming that the Hessian is Lipschitz, the Newton's method satisfies the update 
$||\w_{t+1} - \w^*||_2 \leq c||\w_t - \w^*||^2_2,$
for some constant $c>0$, implying quadratic convergence \cite{nocedal-wright}. 

One shortcoming for the classical Newton's method is that it works only for strongly convex objective functions. In particular,  if $f$ is weakly-convex\footnote{For the sake of clarity, we call a convex function weakly-convex if it is not strongly convex.}, that is, if the Hessian matrix is not positive definite, then the objective function in~\eqref{newton_update} may be unbounded from below. To address this shortcoming, authors in~\cite{newton_mr} recently proposed a variant of  Newton's method, called Newton-Minimum-Residual (Newton-MR). Instead of~\eqref{min_f}, Newton-MR considers the following auxiliary optimization problem: 
$$\min_{\w\in \R^d} ||\nabla f(\w)||^2.$$ 
Note that the  minimizers of this auxiliary problem and \eqref{min_f} are the same when $f(\cdot)$ is convex. Then, the update direction in the $(t+1)$-th iteration is obtained by minimizing the Taylor's expansion of $|| \nabla f(\w_t + \p)||^2$, that is, 
$$\p_{t} = \arg\min_{\w\in \R^d} ||\nabla f(\w_t) +   \h_t\p||^2.$$
The general solution of the above problem is given by  $\p = -[\h_t]^{\dagger}\nabla f(\w_t) + (\I - \h_t[\h_t]^{\dagger})\textbf{q}, ~\forall~ \textbf{q} \in \R^d$, where $[\cdot]^{\dagger}$ is the Moore-Penrose inverse.  Among these, the minimum norm solution is chosen, which gives the update direction in the $t$-th iteration as $\p_t = - \hd_t\nabla f(\w_t)$. Thus, the model update is
\begin{equation}
\label{eq:MR-update}
\w_{t+1} = \w_t + \p_t = \w_t - [\nabla^2f(\w_t)]^{\dagger}\nabla f(\w_t).
\end{equation}
OverSketched Newton considers both of these variants.

\section{OverSketched Newton}\label{sec:osn}

In many applications like machine learning where the training data itself is noisy, using the exact 
Hessian is not necessary. Indeed, many results in the literature prove convergence guarantees for Newton's method when the Hessian is computed approximately using ideas from 
RandNLA for a single machine (e.g. \cite{mert, fred1, fred2, nocedal16}). In particular, these methods perform a form of dimensionality reduction for the Hessian using random matrices, called sketching matrices. Many popular sketching schemes have been proposed in the literature, for example, sub-Gaussian, Hadamard, random row sampling, sparse Johnson-Lindenstrauss, etc. \cite{woodruff_now, Mah-mat-rev_BOOK}. 
Inspired from these works, we present OverSketched Newton, a stochastic second order algorithm for solving---\emph{on serverless systems, in a distributed, straggler-resilient manner}---problems of the form \eqref{min_f}.

\begin{algorithm}[t]
\SetAlgoLined
\SetKwInOut{Input}{Input}
\Input{Matrix $\A\in \R^{t\times s}$, vector $x \in \R^s$, and block size parameter $b$}
\KwResult{$\y = \A\x,$ where $\y \in \R^{s}$ is the product of matrix $\A$ and vector $\x$}
\textbf{Initialization}: Divide $\A$ into $T = t/b$ row-blocks, each of dimension $b\times s$\\
\textbf{Encoding}: Generate coded $\A$, say $\A_c$, in parallel using a 2D product code by arranging the row blocks of $\A$ in a 2D structure of dimension $\sqrt{T}\times \sqrt{T}$ and adding blocks across rows and columns to generate parities; see Fig. 2 in \cite{tavor} for an illustration\\
  \For{$i=1$ to $T + 2\sqrt{T} + 1$}{
1. Worker $W_{i}$ receives the $i$-th row-block of $\A_c$, say $\A_c(i,:)$, and $\x$ from cloud storage\\
2. $W_{i}$ computes $\y(i) = \A(i,:) \x$\\
3. Master receives $\y(i)$ from worker $W_{i}$
}
\textbf{Decoding}: Master checks if it has received results from enough workers to reconstruct $\y$. Once it does, it decodes $\y$ from available results using the peeling decoder
 \caption{Straggler-resilient distributed computation of $\A\x$ using codes}
 \label{algo:coded_computation}
\end{algorithm}


%
%

{\bf Distributed straggler-resilient gradient computation:} 
OverSketched Newton computes the full gradient in each iteration by using tools from error-correcting codes~\cite{kangwook1,tavor}. 
Our key observation is that, for several commonly encountered optimization problems, gradient computation relies on matrix-vector multiplications (see Sec.~\ref{examples_probs} for examples). 
We leverage coded matrix multiplication technique from~\cite{tavor} to perform the large-scale matrix-vector multiplication in a distributed straggler-resilient manner. 
The idea of coded matrix multiplication is explained in Fig.~\ref{fig:kangwook1}; detailed steps are provided in Algorithm \ref{algo:coded_computation}.

\begin{figure}[t]
\begin{minipage}{.49\textwidth}
        \centering
        \includegraphics[scale = 0.39]{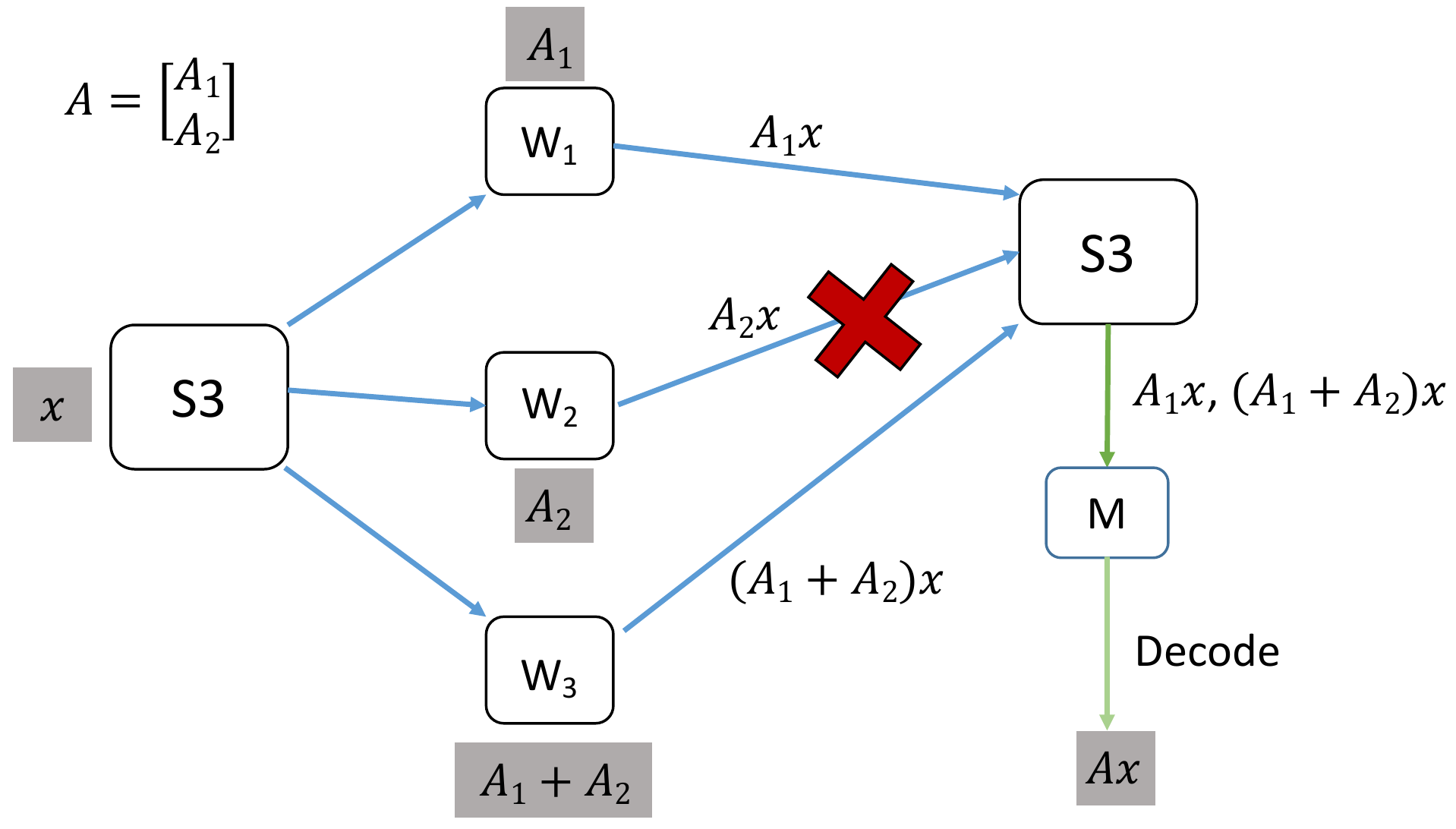}
        \caption{\small \textbf{Coded matrix-vector multiplication}: Matrix $\A$ is divided into 2 row chunks $\A_1$ and $\A_2$. During encoding, redundant chunk $\A_1+\A_2$ is created. Three workers obtain $\A_1$,$\A_2$ and $\A_1+\A_2$ from the cloud storage S3, respectively, and then multiply by $\x$ and write back the result to the cloud. The master $M$ can decode $\A\x$ from the results of any two workers, thus being resilient to one straggler ($W_2$ in this case).}
        \label{fig:kangwook1}
    \end{minipage}
    ~~
    \begin{minipage}{.5\textwidth}
            \vspace{4mm}
        \centering
        \includegraphics[scale=0.46]{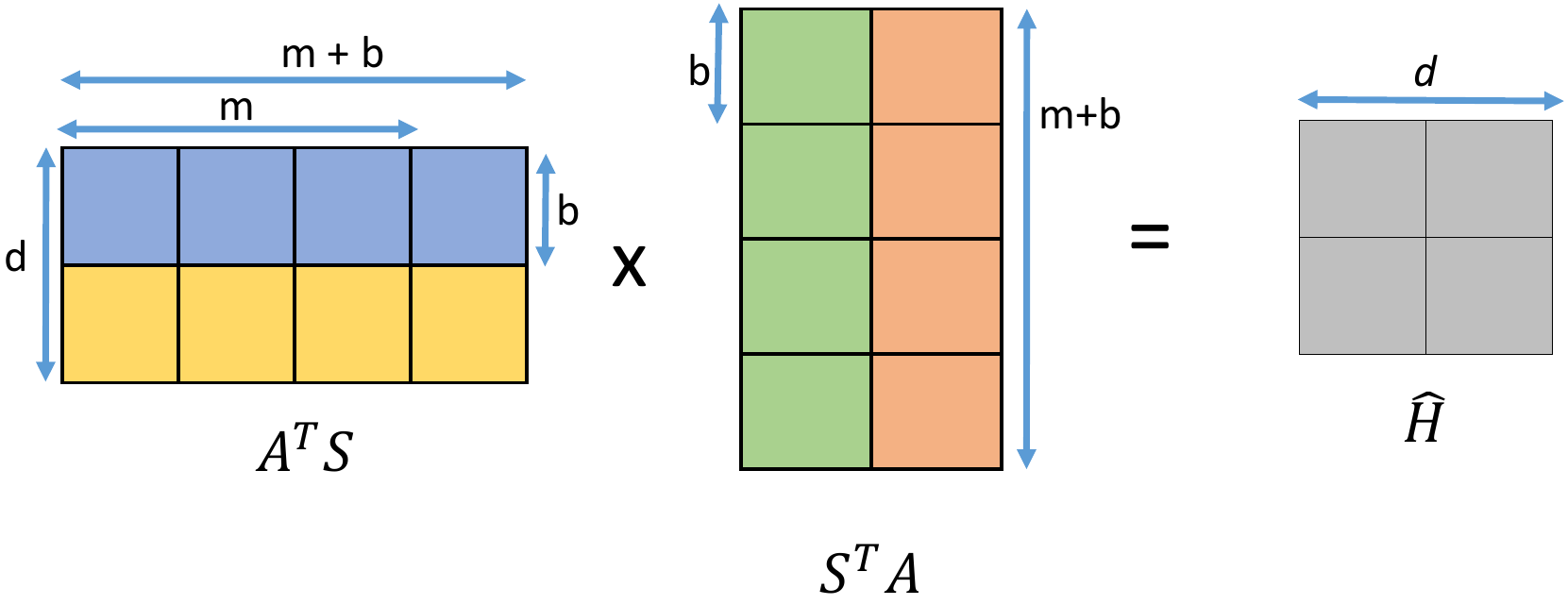}
        \caption{\small \textbf{OverSketch-based approximate Hessian computation:} 
        First, the matrix $\A$---satisfying $\A^T\A = \nabla^2 f(\w_t)$---is sketched in parallel using the sketch in \eqref{sketch_matrix}. Then, each worker receives a block of each of the sketched matrices $\A^T\s$ and $\s^T\A$, multiplies them, and communicates back its results for reduction. During reduction, stragglers can be ignored by the virtue of ``over'' sketching. For example, here the desired sketch dimension $m$ is increased by block-size $b$ for obtaining resiliency against one straggler for each block of $\hh$.}
        \label{fig:oversketch}
    \end{minipage}
    \end{figure}

{\bf Distributed straggler-resilient approximate Hessian computation:} 
For several commonly encountered optimization problems, Hessian computation involves matrix-matrix multiplication for a pair of large matrices (see Sec.~\ref{examples_probs} for several examples).  
For computing the large-scale matrix-matrix multiplication in parallel in serverless systems, we propose to use a straggler-resilient scheme called {\it OverSketch} from~\cite{oversketch}. 
OverSketch does {\it blocked partitioning} of input matrices where each worker works on square blocks of dimension $b$. Hence, it is more communication efficient than existing coding-based straggler mitigation schemes that do na\"ive row-column partition of input matrices \cite{kangwook2, poly_codes}. 
We note that it is well known in HPC that blocked partitioning of input matrices can lead to communication-efficient methods for distributed multiplication \cite{oversketch, 2.5d, summa}.

OverSketch uses a sparse sketching matrix based on Count-Sketch~\cite{woodruff_now}. 
It has similar computational efficiency and accuracy guarantees as that of the Count-Sketch, with two additional properties: it is amenable to distributed implementation; and it is resilient to stragglers. 
More specifically, the OverSketch matrix is constructed as follows~\cite{oversketch}. 

Recall that the Hessian $\nabla^2 f(\cdot) \in\mathbb{R}^{d\times d}$. 
First choose the desired sketch dimension $m$ (which depends on $d$),  block-size $b$ (which depends on the memory of the workers), and straggler tolerance $\zeta > 0$ (which depends on the distributed system). 
Then, define $N = m/b$ and $e = \zeta N$, for some constant $\zeta > 0$. Here $\zeta$ is the fraction of stragglers that we want our algorithm to tolerate.  Thus, $e$ is the maximum number of stragglers per $N+e$ workers that can be tolerated.
The sketch $\s$ is then given by
\begin{equation}\label{sketch_matrix}
\s = \frac{1}{\sqrt{N}}(\s_1, \s_2,\cdots,\s_{N+e}),
\end{equation}
where $\s_i \in \R^{n\times b}$, for all $i\in [1,N+e]$, are i.i.d. Count-Sketch matrices\footnote{Each of the Count-Sketch matrices $\s_i$ is constructed (independently of others) as follows. 
First, for every row $j$, $j\in[n]$, of $\s_i$, independently choose a column $h(j) \in [b]$. 
Then, select a uniformly random element from $\{-1,+1\}$, denoted as $\sigma(i)$. 
Finally, set $\s_{i}(j, h(j)) = \sigma(i)$ and set $\s_{i}(j,l) = 0$ for all $l \ne h(j)$. 
(See~\cite{woodruff_now, oversketch} for details.)} with sketch dimension $b$. Note that $\s\in \R^{n\times (m+eb)}$, where $m = Nb$ is the required sketch dimension and $e$ is the over-provisioning parameter to provide resiliency against $e$ stragglers per $N+e$ workers. We leverage the straggler resiliency of OverSketch to obtain the sketched Hessian in a distributed straggler-resilient manner.  An illustration of OverSketch is provided in Fig.~\ref{fig:oversketch}; see Algorithm \ref{algo:OverSketch} for details.


\begin{algorithm}[t]
\SetAlgoLined
\SetKwInOut{Input}{Input}
\Input{Matrices $\A\in \R^{n\times d}$, required sketch dimension $m$, straggler tolerance $e$, block-size $b$. Define $N = m/b$}
\KwResult{$\hh \approx \A^T\times \A$ }
\textbf{Sketching}: Use sketch in Eq. \eqref{sketch_matrix} to obtain $\tA = \s^T\A$ distributedly (see Algorithm 5 in \cite{oversketch} for details)\\
\textbf{Block partitioning}: Divide $\tA$ into $(N + e)\times d/b$ matrix of $b\times b$ blocks\\
\textbf{Computation phase}: Each worker takes a block of $\tA$ and $\tA^T$ each and multiplies them. This step invokes $(N+e)d^2/b^2$ workers, where $N+e$ workers compute one block of $\hh$\\
\textbf{Termination}: Stop computation when any $N$ out of $N+e$ workers return their results for each block of $\hh$\\
\textbf{Reduction phase}: Invoke $d^2/b^2$ workers to aggregate results during the computation phase, where each worker will calculate one block of $\hh$
 \caption{Approximate Hessian calculation on serverless systems using OverSketch}
 \label{algo:OverSketch}
\end{algorithm}



\begin{algorithm}[t]
\SetAlgoLined
\SetKwInOut{Input}{Input}
\Input{Convex function $f$; Initial iterate $\w_0\in\mathbb{R}^d$; Line search parameter $0 < \beta \leq 1/2$; Number of iterations $T$}
  \For{$t=1$ to $T$}{
Compute full gradient $\g_t$ in a distributed fashion using Algorithm \ref{algo:coded_computation} \\
Compute sketched Hessian matrix $\hat{\mathbf{H}}_t$ in a distributed fashion using Algorithm \ref{algo:OverSketch}\\
\eIf{$f$ is strongly-convex}{
Compute the update direction at the master as: $\p_t = - [\hat{\mathbf{H}}_t]^{-1}\nabla f(\w_t)$\\
Compute step-size $\alpha_t$ satisfying the line-search condition~\eqref{eq:line-search} in a distributed fashion\\
}
{
Compute the update direction at the master as: $\p_t = - [\hat{\mathbf{H}}_t]^{\dagger}\nabla f(\w_t)$\\
Find step-size $\alpha_t$ satisfying the line-search condition~\eqref{eq:line-search2} in a distributed fashion
}
Compute the model update $\w_{t+1} = \w_t + \alpha_t \p_t$ at the master
}
 \caption{OverSketched Newton in a nutshell}
 \label{algo:OSN}
\end{algorithm}

{\bf Model update:} Let $\hh_t= \A_t^T\s_t\s_t^T\A_t$, where $\A_t$ is the square root of the Hessian $\nabla^2f(\w_t)$, and $\s_t$ is an independent realization of \eqref{sketch_matrix} at the $t$-th iteration. For strongly-convex functions, the update direction is $\p_t = -\hh_t^{-1}\nabla f(\w_t)$.  We use line-search to choose the step-size, that is, find
\begin{align}\label{eq:line-search}
\alpha_t = \max_{\alpha \leq 1} ~\alpha
\text{~~~such that~~~}
 f(\w_t+ \alpha\p_t) \leq f(\w_t) + \alpha\beta\p_t^T\nabla f(\w_t),
\end{align}
for some constant $\beta \in (0,1/2]$. For weakly-convex functions, the update direction (inspired by Newton-MR~\cite{newton_mr}) is $\p_t = -\hh_t^{\dagger}\nabla f(\w_t)$, where  $\ddh_t$ is the Moore-Penrose inverse of $\hh_t$. To find the update $\w_{t+1}$, we find the right step-size $\alpha_t$ using line-search in \eqref{eq:line-search}, 
but with $f(\cdot)$ replaced by $||\nabla f(\cdot)||^2$ and $\nabla f(\w_t)$ replaced by $2\hh_t\nabla f(\w_t)$, according to the objective in $||\nabla f(\cdot)||^2$. 
More specifically, for some constant $\beta \in (0,1/2]$, 
\begin{align}\label{eq:line-search2}
\alpha_t = \max_{\alpha\leq 1} ~\alpha \text{~~~such that~~~} 
 ||\nabla f(\w_t+ \alpha\p_t)||^2 \leq ||\nabla f(\w_t)||^2 + 2\alpha\beta\p_t^T\hh_t\nabla f(\w_t).
\end{align}
Note that for OverSketched Newton, we use $\hh_t$ in the line-search since the exact Hessian is not available.  
The update in the $t$-th iteration in both cases is given by 
$$\w_{t+1} = \w_t + \alpha_t\p_t.$$ 
Note that 
\eqref{eq:line-search} 
line-search can be solved approximately in single machine systems using Armijo backtracking line search \cite{boyd,nocedal-wright}. OverSketched Newton is concisely described in Algorithm \ref{algo:OSN}. 
In Section \ref{line-search}, we describe how to implement distributed line-search in serverless systems when the data is stored in the cloud. 
Next, we prove convergence guarantees for OverSketched Newton that uses the sketch matrix in \eqref{sketch_matrix} and full gradient for approximate Hessian computation. 

\subsection{Convergence Guarantees}
 
First, we focus our attention to strongly convex functions. We consider the following assumptions. We note that these assumptions are standard for analyzing approximate Newton methods, (e.g., see ~\cite{mert,fred1,fred2}.

{\bf Assumptions:} 
\vspace{-4mm}
\begin{enumerate}
\item $f$ is twice-differentiable;
\item $f$ is $k$-strongly convex ($k > 0$), that is,
$$\nabla^2f(\w) \succeq k\I ;$$
\item $f$ is $M$-smooth ($k\leq M < \infty$), that is, 
$$\nabla^2f(\w) \preceq M\I;$${}
\item the Hessian is $L$-Lipschitz continuous, that is, for any $\pmb\Delta\in \R^d$
$$||\nabla^2 f(\w +\Del) - \nabla^2f(\w)||_{2} \leq L||\Del||_2,$$
where $||\cdot||_{2}$ is the spectral norm for matrices. 
\end{enumerate}

We first prove the following ``global'' convergence guarantee which shows that OverSketched Newton would converge from any random initialization of $\w_0\in \R^d$ with high probability. 

\begin{theorem}[\textbf{Global convergence for strongly-convex $f$}]\label{global_thm}
Consider Assumptions 1, 2, and 3 and step-size $\alpha_t$ given by Eq. \eqref{eq:line-search}. Let $\w^*$ be the optimal solution of \eqref{min_f}. 
Let $\epsilon$ and $\mu$ be positive constants. Then, 
using the sketch in \eqref{sketch_matrix} with a sketch dimension $Nb+eb = \Omega(\frac{d^{1+\mu}}{\epsilon^2})$ and the number of column-blocks $N +e= \Theta_\mu(1/\epsilon)$,
the updates for OverSketched Newton, for any $\w_t \in \R^d$, satisfy
\begin{equation*}
f(\w_{t+1}) - f(\w^*) \leq (1-\rho)(f(\w_{t}) - f(\w^*)),
\end{equation*}
with probability at least $1 - 1/d^\tau$, where $\rho = \frac{2\alpha_t\beta k}{M(1+\epsilon)}$ and $\tau > 0$ is a constant depending on $\mu$ and constants in $\Omega(\cdot)$ and $\Theta(\cdot)$. Moreover, $\alpha_t$ satisfies $\alpha_t\geq \frac{2(1-\beta)(1-\epsilon)k}{M}$. 
\end{theorem}
\begin{proof}
See Section \ref{proof_global_conv}. 
\end{proof}

Theorem \ref{global_thm} guarantees the global convergence of OverSketched Newton starting with any initial estimate $\w_0 \in \R^d$ to the optimal solution $\w^*$ with at least a linear rate. 

Next,
we can also prove an additional ``local'' convergence guarantee for OverSketched Newton, under the assumption that $\w_0$ is sufficiently close to $\w^*$. 
\begin{theorem}[\textbf{Local convergence for strongly-convex $f$}]\label{convergence_thm} 
Consider Assumptions 1, 2, and 4 and step-size $\alpha_t = 1$. Let $\w^*$ be the optimal solution of \eqref{min_f} and $\gamma$ and $\beta$ be the minimum and maximum eigenvalues of $\nabla^2f(\w^*)$, respectively. 
Let $\epsilon \in (0,\gamma/(8\beta)]$ and $\mu > 0$. Then, using the sketch in \eqref{sketch_matrix} with a sketch dimension $Nb + eb = \Omega(\frac{d^{1+\mu}}{\epsilon^2})$ 
and the number of column-blocks $N + e = \Theta_\mu(1/\epsilon)$, the updates for OverSketched Newton, with initialization $\w_0$ such that $||\w_0 - \w^*||_2 \leq \frac{\gamma}{8L}$, follow
\begin{equation*}\label{conv_eqn}
||\w_{t+1} - \w^*||_2 \leq \frac{25L}{8\gamma}||\w_t - \w^*||^2_2 + \frac{5\epsilon\beta}{\gamma}||\w_t - \w^*||_2 ~~~\text{for}~ t = 1,2,\cdots, T,
\end{equation*}
with probability at least $1 -T/d^\tau$, where $\tau > 0$ is a constant depending on $\mu$ and constants in $\Omega(\cdot)$ and $\Theta(\cdot)$. 
\end{theorem}
\begin{proof}
See Section \ref{proof_local_conv}.
\end{proof}
Theorem \ref{convergence_thm} implies that the convergence is linear-quadratic in error $\Delta_t = \w_t - \w^*$. 
Initially, when $||\Delta_t||_2$ is large, the first term of the RHS will dominate and the convergence will be quadratic, that is, $||\Delta_{t+1}||_2 \lesssim \frac{25L}{8\gamma}||\Delta_t||^2_2$. In later stages, when $||\w_t - \w^*||_2$ becomes sufficiently small, the second term of RHS will start to dominate and the convergence will be linear, that is, $||\Delta_{t+1}||_2 \lesssim \frac{5\epsilon\beta}{\gamma}||\Delta_t||_2$. At this stage, the sketch dimension can be increased to reduce $\epsilon$ to diminish the effect of the linear term and improve the convergence~rate in practice. Note that, for second order methods, the number of iterations $T$ is in the order of tens in general while the number of features $d$ is typically in thousands. Hence, the probability of failure is generally small (and can be made negligible by choosing $\tau$ appropriately). 



Though the works \cite{mert, fred1, fred2,nocedal16,nocedal17} also prove convergence guarantees for approximate Hessian-based optimization, no convergence results exist for the OverSketch matrix in Eq. \eqref{sketch_matrix} to the best of our knowledge. OverSketch has many nice properties like sparsity, input obliviousness, and amenability to distributed implementation, and our convergence guarantees take into account the block-size $b$ (that captures the amount of communication at each worker) and the number of stragglers $e$, both of which are a property of the distributed system. On the other hand, algorithms in \cite{mert, fred1, fred2,nocedal16,nocedal17}  are tailored to run on a single machine.

Next, we consider the case of weakly-convex functions. 
For this case, we consider two more assumptions on the Hessian matrix, similar to \cite{newton_mr}. These assumptions are a relaxation of the strongly-convex case.

{\bf Assumptions:} 
\vspace{-4mm}
\begin{enumerate}
  \setcounter{enumi}{4}
  \item There exists some $\eta > 0$ such that, $\forall~\w\in \R^d$, 
  $$||(\nabla^2 f(\w))^{\dagger}||_2 \leq1/\eta.$$ 
  This assumption establishes regularity on the pseudo-inverse of $\nabla^2 f(\x)$. 
It also implies that
$||\nabla^2 f(\w)\p|| \geq \eta||\p|| ~\forall ~p \in\text{ Range}(\nabla^2 f(\w))$, that is, the minimum `non-zero' eigenvalue of $\nabla^2 f(\w)$ is lower bounded by $\eta$;  just as in the $k$-strongly convex case, the smallest eigenvalue is greater than $k$.

  \item Let $\U \in \R^{d\times d}$ be any arbitrary orthogonal basis for $\text{Range} (\nabla^2 f(\w))$, there exists $0 <\nu \leq 1$, such that, 
  $$||\bf{U}^T\nabla f(\w)||^2 \geq \nu ||\nabla f(\w)||^2 ~~\forall ~~\w\in\R^d.$$
  This assumption ensures that there is always a non-zero component of the gradient in the subspace spanned by the Hessian, and, thus, ensures that the model update $-\hh_t^{\dagger}\nabla f(\w_t)$ will not be zero.
\end{enumerate} 
Note that the above assumptions are always satisfied by strongly-convex functions. Next, we prove global convergence of OverSketched Newton when the objective is weakly-convex.


\begin{theorem}[\textbf{Global convergence for weakly-convex $f$}]\label{global_thm_weakly_convex}
Consider Assumptions 1,3,4,5 and 6 and step-size $\alpha_t$ given by Eq. \eqref{eq:line-search2}. 
Let $\epsilon \in \left(0, \frac{(1-\beta)\nu\eta}{2M}\right]$ and $\mu > 0$. 
Then, using an OverSketch matrix with a sketch dimension $Nb + eb = \Omega(\frac{d^{1+\mu}}{\epsilon^2})$ and the number of column-blocks $N +e= \Theta_\mu(1/\epsilon)$, the updates for OverSketched Newton, for any $\w_t \in \R^d$, satisfy
\vspace{-2mm}
\begin{align*}
||\nabla f(\w_{t+1})||^2 \leq \bigg(1 - 2\beta\alpha\nu\frac{(1-\epsilon)\eta}{M(1+\epsilon)}\bigg)||\nabla f(\w_t)||^2,
\end{align*}
with probability at least $1 - 1/d^\tau$, where 
$\alpha = \frac{\eta}{2Q}\big[(1-\beta)\nu\eta - 2\epsilon M\big], Q = (L||\nabla f(\w_0)|| + M^2),$ $\w_0$ is the initial iterate of the algorithm and $\tau > 0$ is a constant depending on $\mu$ and constants in $\Omega(\cdot)$ and $\Theta(\cdot)$.
\end{theorem}
\begin{proof}
See Section \ref{proof_global_conv_weakly_convex}.
\end{proof}



Even though we present the above guarantees for the sketch matrix in Eq. \eqref{sketch_matrix}, our analysis is valid for any sketch that satisfies the {\it subspace embedding} property (Lemma~\ref{lemma1}; see \cite{woodruff_now} for details on subspace embedding property of sketches).
To the best of our knowledge, this is the first work to prove the convergence guarantees for weakly-convex functions when the Hessian is calculated approximately using sketching techniques. Later, authors in \cite{liu_fred_app_hess_newtonMR} extended the analysis to the case of general Hessian perturbations with additional assumptions on the type of perturbation.


\subsection{Distributed Line Search}\label{line-search}

Here, we describe a line-search procedure for distributed serverless optimization, which is inspired by the line-search method from \cite{giant} for serverful systems. 
To solve for the step-size $\alpha_t$ as described in the optimization problem in \eqref{eq:line-search}, we set $\beta=0.1$ and choose a candidate set $\mathcal{S} = \{4^0, 4^1, \cdots, 4^{-5}\}$. After the master calculates the descent direction $\p_t$ in the $t$-th iteration, the $i$-th worker calculates $f_i(\w_t+\alpha\p_t)$ for all values of $\alpha$ in the candidate set $\mathcal{S}$, where $f_i(\cdot)$ depends on the local data available at the $i$-th worker and $f(\cdot) = \sum_i f_i(\cdot)$\footnote{For the weakly-convex case, the workers calculate $\nabla f_i(\cdot)$ instead of $f_i(\cdot)$, and the master calculates $||\nabla f(\cdot)||^2$ instead of $f(\cdot)$.}. 

The master then sums the results from workers to obtain $f (\w_t+\alpha\p_t)$ for all values of $\alpha$ in $\mathcal{S}$ and finds the largest $\alpha$ that satisfies the Armijo condition in \eqref{eq:line-search}\footnote{Note that codes can be used to mitigate stragglers during distributed line-search in a manner similar to the gradient computation phase.}. 
Note that line search requires an additional round of communication where the master communicates $\p_t$ to the workers through cloud and the workers send back the function values $f_i(\cdot)$. Finally, the master finds the best step-size from set $\mathcal{S}$ and finds the model estimate $\w_{t+1}$. 

\section{OverSketched Newton on Serverless Systems: Examples}\label{examples_probs}
Here, we describe several examples where our general approach can be applied.

\subsection{Logistic Regression using OverSketched Newton}
\label{sec:logistic}

The optimization problem for supervised learning using Logistic Regression takes the form
\begin{equation}
\min_{\w\in \R^d}~ \bigg\{f(\w) = \frac{1}{n}\sum_{i=1}^n\log(1 + e^{-y_i\w^T\x_i}) + \frac{\lambda}{2}\|\w\|^2_2 \bigg\}. \label{LReg}
\end{equation}
Here, $\x_1,\cdots, \x_n \in \R^{d\times 1}$ and $y_1,\cdots,y_n \in \R$ are training sample vectors and labels, respectively. The goal is to learn the feature vector $\w^* \in \R^{d\times 1}$. 
Let $\mathbf{X} = [\x_1,\x_2,\cdots,\x_n]\in \R^{d\times n}$ and $\y = [y_1,\cdots,y_n]\in \R^{n\times 1}$ be the example and label matrices, respectively.
The gradient for the problem in \eqref{LReg} is given by
$$\nabla f(\w) =  \frac{1}{n}\sum_{i=1}^{n}\frac{-y_i\x_i}{1 + e^{y_i\w_i^T\x_i}} + \lambda\w.$$ 

Calculation of $\nabla f(\w)$ involves two matrix-vector products, $\pmb{\alpha} = \mathbf{X}^T\w$ and $\nabla f(\w) = \frac{1}{n}\mathbf{X}\pmb{\beta} + \lambda\w$, where $\beta_i = \frac{-y_i}{1 + e^{y_i\alpha_i}} ~\forall~i\in[1,\cdots,n]$. When the example matrix is large, these matrix-vector products are performed distributedly using codes.
Faster convergence is obtained by second-order methods which will additionally compute the Hessian $\mathbf{H} = \frac{1}{n}\mathbf{X}\pmb\Lambda\mathbf{X}^T + \lambda\I_d$, where $\pmb\Lambda$ is a diagonal matrix with entries given by $\Lambda(i,i) = \frac{e^{y_i\alpha_i}}{(1+e^{y_i\alpha_i})^2}$. The product $\mathbf{X}\pmb\Lambda\mathbf{X}^T$ is computed approximately in a distributed straggler-resilient manner using the sketch matrix in~\eqref{sketch_matrix}. 
Using the result of distributed multiplication, the Hessian matrix $\textbf{H}$ is calculated at the master and the model is updated as $\w_{t+1} = \w_t - \textbf{H}^{-1}\nabla f(\w_t)$. 
In practice, efficient algorithm like conjugate gradient, that provide a good estimate in a small number of iterations, can be used locally at the master to solve for $\w_{t+1}$ \cite{cg}.\footnote{Note that here we have assumed that the number of features is small enough to perform the model update locally at the master. This is not necessary, and straggler resilient schemes, such as in \cite{local_codes}, can be used to perform distributed conjugate gradient in serverless systems.}

 \begin{algorithm}[t]
\SetAlgoLined
\textbf{Input Data} (stored in cloud storage):
Example Matrix $\X\in \R^{d\times n}$ and vector $\y\in \R^{n\times 1}$ (stored in cloud storage), regularization parameter $\lambda$, number of iterations $T$, Sketch $\s$ as defined in Eq. \eqref{sketch_matrix}\\
\textbf{Initialization}: Define $\w^1 = \mathbf{0}^{d\times1}, \pmb{\beta} = \mathbf{0}^{n\times1}, \pmb\gamma = \mathbf{0}^{n\times 1}$, Encode $\textbf{X}$ and $\textbf{X}^T$ as described in Algorithm~\ref{algo:coded_computation} \\ 
 \For{$t=1$ to $T$}{
 $\pmb{\alpha} = \X\w^t$ \tcp*{Compute in parallel using Algorithm \ref{algo:coded_computation}}
 \For{$i=1$ to $n$}{
 $\beta_i = \frac{-y_i}{1 + e^{y_i\alpha_i}};$
 }
$\mathbf{g} =  \mathbf{X}^T\pmb{\beta}$ \tcp*{Compute in parallel using Algorithm \ref{algo:coded_computation}}
$\nabla f(\w^t) = \mathbf{g} + \lambda\w^t;$\\
\For{$i=1$ to $n$}{
$\gamma(i) = \frac{e^{y_i\alpha_i}}{(1+e^{y_i\alpha_i})^2};$
}
$\textbf{A} = \sqrt{diag(\pmb\gamma)}\mathbf{X}^T$\\
$\hat{\mathbf{H}} = \mathbf{A}^T\s\s^T\mathbf{A}$ 
\tcp*{Compute in parallel using Algorithm \ref{algo:OverSketch}}
$\mathbf{H} = \frac{1}{n}\hat{\textbf{H}} + \lambda\I_d$;\\
$\w^{t+1} = \w^t - \textbf{H}^{-1}\nabla f(\w^t)$;
}
\KwResult{$\w^* = \w_{T+1}$}
 \caption{OverSketched Newton: Logistic Regression for Serverless Computing}
 \label{algo:logistic_lambda}
\end{algorithm}

We provide a detailed description of OverSketched Newton for large-scale logistic regression for serverless systems in Algorithm \ref{algo:logistic_lambda}. Steps 4, 8, and 14 of the algorithm are computed in parallel on AWS Lambda. All other steps are simple vector operations that can be performed locally at the master, for instance, the user's laptop. Steps 4 and 8 are executed in a straggler-resilient fashion using the coding scheme in \cite{tavor}, as illustrated in Fig. \ref{fig:stragglers} and described in detail in Algorithm \ref{algo:coded_computation}.

We use the coding scheme in \cite{tavor} since the encoding can be implemented in parallel and requires less communication per worker compared to the other schemes, for example schemes in \cite{kangwook1, poly_codes}, that use Maximum Distance Separable (MDS) codes. Moreover, the decoding scheme takes linear time and is applicable on real-valued matrices. 
Note that since the example matrix $\X$ is constant in this example, the encoding of $\X$ is done only once before starting the optimization algorithm.
Thus, the encoding cost can be amortized over iterations. 
Moreover, decoding over the resultant product vector requires negligible time and space, even when $n$ is scaling into the millions. 

The same is, however, not true for the matrix multiplication for Hessian calculation (step 14 of Algorithm \ref{algo:logistic_lambda}), as the matrix $\textbf{A}$ changes in each iteration, thus encoding costs will be incurred in every iteration if error-correcting codes are used. Moreover, encoding and decoding a huge matrix stored in the cloud incurs heavy communication cost and becomes prohibitive. 
Motivated by this, we use OverSketch in step 14, as described in Algorithm \ref{algo:OverSketch}, to calculate an approximate matrix multiplication, and hence the Hessian, efficiently in serverless systems with inbuilt straggler resiliency.\footnote{We also evaluate the exact Hessian-based algorithm with speculative execution, i.e., recomputing the straggling jobs, and compare it with OverSketched Newton in Sec.~\ref{sec:experiments}.} 

\subsection{Softmax Regression using OverSketched Newton}
\label{sec:softmax}
We take unregulairzed softmax regression as an illustrative example for the weakly convex case. The goal is to find the weight matrix $\mathbf{W}=\left[\mathbf{w}_{1}, \cdots, \mathbf{w}_{K}\right]$ that fit the training data $\X \in \R^{d\times N}$ and $\y\in \R^{K\times N}$. Here $\w_i\in \R^d$ represesents the weight vector for the $k$-th class for all $i \in [1, K]$ and $K$ is the total number of classes. Hence, the resultant feature dimension for softmax regression is $dK$. 
The optimization problem is of the form
\begin{align}\label{softmax_objective}
    f(\W) = \sum_{n=1}^{N}\left[\sum_{k=1}^{K} y_{kn} \mathbf{w}_{k}^{T} \mathbf{x}_{n}-\log \sum_{l=1}^{K} \exp {\left(\mathbf{w}_{l}^{T} \mathbf{x}_{n}\right)}\right].
\end{align}
The gradient vector for the $i$-th class is given by
\begin{align}
\nabla f_i(\W) = \sum_{n=1}^{N}\left[\frac{\exp \left(\mathbf{w}_{i}^{T} \mathbf{x}_{n}\right)}{\sum_{l=1}^{K} \exp \left(\mathbf{w}_{l}^{T} \mathbf{x}_{n}\right)} -  y_{i n} \right]\mathbf{x}_{n} ~~~\forall ~i\in [1,k],
\end{align}
which can be written as matrix products $\pmb{\alpha_i} = \X^T\w_i$ and $\nabla f_i(\W) = \X\pmb\beta_i$, where the entries of $\pmb\beta_i \in \R^N$ are given by $\beta_{in} = \left(\frac{\exp (\alpha_{in})}{\sum_{l=1}^{K} \exp (\alpha_{ln})} - y_{in}\right)$. Thus, the full gradient matrix is given by 
$\nabla f(\W) = \X\pmb{\beta}$ where the entries of $\pmb{\beta} \in \R^{N\times K}$ are dependent on $\pmb{\alpha}\in \R^{N\times K}$ as above and the matrix $\pmb{\alpha}$ is given by $\pmb{\alpha} = \X^T\W$. We assume that the number of classes $K$ is small enough such that tall matrices $\pmb\alpha$ and $\pmb\beta$ are small enough for the master to do local calculations on them.

Since the effective number of features is $d\times K$, the Hessian matrix is of dimension $dK\times dK$. The $(i,j)$-th component of the Hessian, say $\h_{ij}$, is 
\begin{align}
\h_{ij}(\W) =\frac{d}{d \mathbf{w}_{j}} \nabla f_{i}(\mathbf{W})=\frac{d}{d \mathbf{w}_{j}} \mathbf{X}\pmb{\beta_i}=\mathbf{X} \frac{d}{d \mathbf{w}_{j}} \pmb\beta_i=\mathbf{X} \mathbf{Z_{ij}}{\X^T}
\end{align}
where $\mathbf Z_{ij} \in \R^{N\times N}$ is a diagonal matrix whose $n$-th diagonal entry is
\begin{align}
{Z}_{ij}(n) = \frac{\exp (\alpha_{in})}{\sum_{l=1}^{K} \exp (\alpha_{ln})}\left(\mathbb{I}(i=j) - \frac{\exp (\alpha_{jn})}{\sum_{l=1}^{K} \exp (\alpha_{ln})}\right)~\forall~n\in [1,N],
\end{align}
where $\mathbb{I}(\cdot)$ is the indicator function and $\pmb\alpha = \X\W$ was defined above. The full Hessian matrix is obtained by putting together all such $\h_{ij}$'s in a $dK\times dK$ matrix and can be expressed in a matrix-matrix multiplication form as
\begin{align}
\nabla^2 f(\W) = \left[ \begin{array}{ccc}{\mathbf{H}_{11}} & {\cdots} & {\mathbf{H}_{1 K}} \\ {\vdots} & {\ddots} & {\vdots} \\ {\mathbf{H}_{K 1}} & {\cdots} & {\mathbf{H}_{K K}}\end{array}\right]
= \left[ \begin{array}{ccc}{\X\mathbf{Z}_{11}}\X^T & {\cdots} & \X{\mathbf{Z}_{1 K}}\X^T \\ {\vdots} & {\ddots} & {\vdots} \\ \X{\mathbf{Z}_{K 1}}\X^T & {\cdots} & \X{\mathbf{Z}_{K K}\X^T}\end{array}\right] = \bar\X \bar\Z \bar\X^T,
\end{align}
where $\bar\X \in\R^{dK\times NK}$ is a block diagonal matrix that contains $\X$ in the diagonal blocks and $\bar{\Z} \in \R^{NK\times NK}$ is formed by stacking all the $\Z_{ij}$'s for $i,j \in [1,K]$.
In OverSketched Newton, we compute this multiplication using sketching in serverless systems for efficiency and resiliency to stragglers. Assuming $d\times K$ is small enough, the master can then calculate the update $\p_t$ using efficient algorithms such the minimum-residual method \cite{mr,newton_mr}.

\subsection{Other Example Problems}
In this section, we describe several other commonly encountered optimization problems that can be solved using OverSketched Newton. 

\textbf{Ridge Regularized Linear Regression}: The optimization problem is
\begin{equation}
\min_{\w\in \R^d} ~ \frac{1}{2n}||\X^T\w - \y||_2^2+ \frac{\lambda}{2}\|\w\|^2_2.\label{linear_reg}    
\end{equation}
The gradient in this case can be written as $\frac{1}{n}\X(\pmb\beta - \y) + \lambda\w$, where $\pmb\beta = \X^T\w$, where the training matrix $\X$ and label vector $\y$ were defined previously. The Hessian is given by $\nabla^2 f(\w) = \X\X^T + \lambda\I$. For $n\gg d$, this can be computed approximately using the sketch matrix in \eqref{sketch_matrix}.
 
\textbf{Linear programming via interior point methods}: The following linear program can be solved using OverSketched Newton
\begin{align}\label{lp}
\underset{\A\x\leq \bb}{\text{minimize}}~ \mathbf{c}^T\x,
\end{align}
where $\x \in \R^{m\times 1}, \mathbf{c} \in \R^{m\times 1}, \bb \in \R^{n\times 1}$ and $\A\in \R^{n\times m}$ is the constraint matrix with $n>m$. 
In algorithms based on interior point methods, the following sequence of problems using Newton's method
\begin{equation}\label{int_point}
\min_{\x\in \R^m}  f(\x) = \min_{\x\in \R^m} \left(\tau \textbf{c}^T\x - \sum_{i=1}^n\log(b_i - \ba_i\x)\right),
\end{equation}
where $\ba_i$ is the $i$-th row of $\A$, $\tau$ is increased geometrically such that when $\tau$ is very large, the logarithmic term does not affect the objective value and serves its purpose of keeping all intermediates solution inside the constraint region. The update in the $t$-th iteration is given by 
$\x_{t+1} = \x_t - (\nabla^2f(\x_t))^{-1}\nabla f(\x_t)$,
where $\x_t$ is the estimate of the solution in the $t$-th iteration. The gradient can be written as $\nabla f(\x) = \tau\textbf{c} + \A^T\pmb\beta$ where $\beta_i = 1/(b_i - \alpha_i)$ and $\pmb\alpha = \A\x$.

The Hessian for the objective in \eqref{int_point} is given by
\begin{eqnarray}
\nabla^2f(\x) = \A^T\text{diag}\frac{1}{(b_i - \alpha_i)^2}\A.
\end{eqnarray} 
The square root of the Hessian is given by $\nabla^2f(\x)^{1/2} = \text{diag}\frac{1}{|bi - \alpha_i|}\A$. The  computation of Hessian requires $O(nm^2)$ time and is the bottleneck in each iteration. Thus, we can use sketching to mitigate stragglers while evaluating the Hessian efficiently, i.e. $\nabla^2f(\x) \approx (\s\nabla^2f(\x)^{1/2})^T\times (\s\nabla^2f(\x)^{1/2})$, where $\s$ is the OverSketch matrix defined in \eqref{sketch_matrix}. 

\textbf{Lasso Regularized Linear Regression}: The optimization problem takes the following form
\begin{equation}\label{lasso}
\min_{\w\in\R^d}~\frac{1}{2}||\X\w - \y||_2^2 + \lambda ||\w||_1,
\end{equation}
where $\X\in \R^{n\times d}$ is the measurement matrix, the vector $\y\in \R^n$ contains the measurements, $\lambda\geq 0$ and $d\gg n$. To solve \eqref{lasso}, we consider its dual variation
$$\min_{\substack{||\X^T\z||_\infty \leq \lambda, \z\in \R^n}} \frac{1}{2}||\y - \z||_2^2,$$
which is amenable to interior point methods and can be solved by optimizing the following sequence of problems where $\tau$ is increased geometrically
$$
\min_{\z}f(\z) = \min_{\z} \Big(\frac{\tau}{2}||\y - \z||_2^2 - \sum_{j=1}^d\log(\lambda - \x_j^T\z)- \sum_{j=1}^d(\lambda + \x_j^T\z)\Big),$$
where $\x_j$ is the $j$-th column of $\X$.
The gradient can be expressed in few matrix-vector multiplications as $\nabla f(\z) = \tau(\z - \y) + \X(\pmb\beta- \pmb\gamma),$ where $\beta_i = 1/(\lambda - \alpha_i)$,  $\gamma_i = 1/(\lambda + \alpha_i)$,  and $\pmb\alpha = \X^T\z$. Similarly, the Hessian can be written as  $\nabla^2 f(\z) = \tau\I + \X\pmb\Lambda\X^T$, where $\pmb\Lambda$ is a diagonal matrix whose entries are given by $\Lambda_{ii} = 1/(\lambda - \alpha_i)^2 + 1/(\lambda + \alpha_i)^2~\forall~ i\in [1,n]$.


Other common problems where OverSketched Newton is applicable include Linear Regression, Support Vector Machines (SVMs), Semidefinite programs, etc.


\section{Experimental Results}
\label{sec:experiments}
In this section, we evaluate OverSketched Newton on AWS Lambda using real-world and synthetic datasets, and we compare it with state-of-the-art distributed optimization algorithms\footnote{A working implementation of OverSketched Newton is available at https://github.com/vvipgupta/OverSketchedNewton}. We use the serverless computing framework, Pywren \cite{pywren}.
Our experiments are focused on logistic and softmax regression, which are popular supervised learning problems, but they can be reproduced for other problems described in Section \ref{examples_probs}.
We present experiments on the following datasets:

 \begin{table}[h] 
 \centering 
 \begin{tabular}{cccc} 
\multicolumn{1}{c}{\bf Dataset}  &
\multicolumn{1}{c}{\bf Training Samples} &
\multicolumn{1}{c}{\bf Features} &
\multicolumn{1}{c}{\bf Testing samples}  \\
 \hline  
 Synthetic & $300,000$ & $3000$ & $100,000$\\
 EPSILON & $400,000$ & $2000$ & $100,000$ \\ 
 WEBPAGE & $48,000$ & $300$ & $15,000$ \\ 
 a9a & $32,000$ & $123$ & $16,000$ \\ 
 EMNIST & $240,000$ & $7840$ & $40,000$   \\ 
 \hline 
 \end{tabular} 
 \end{table}


 \begin{figure}[t]
 \centering
 \includegraphics[width=0.5\textwidth]{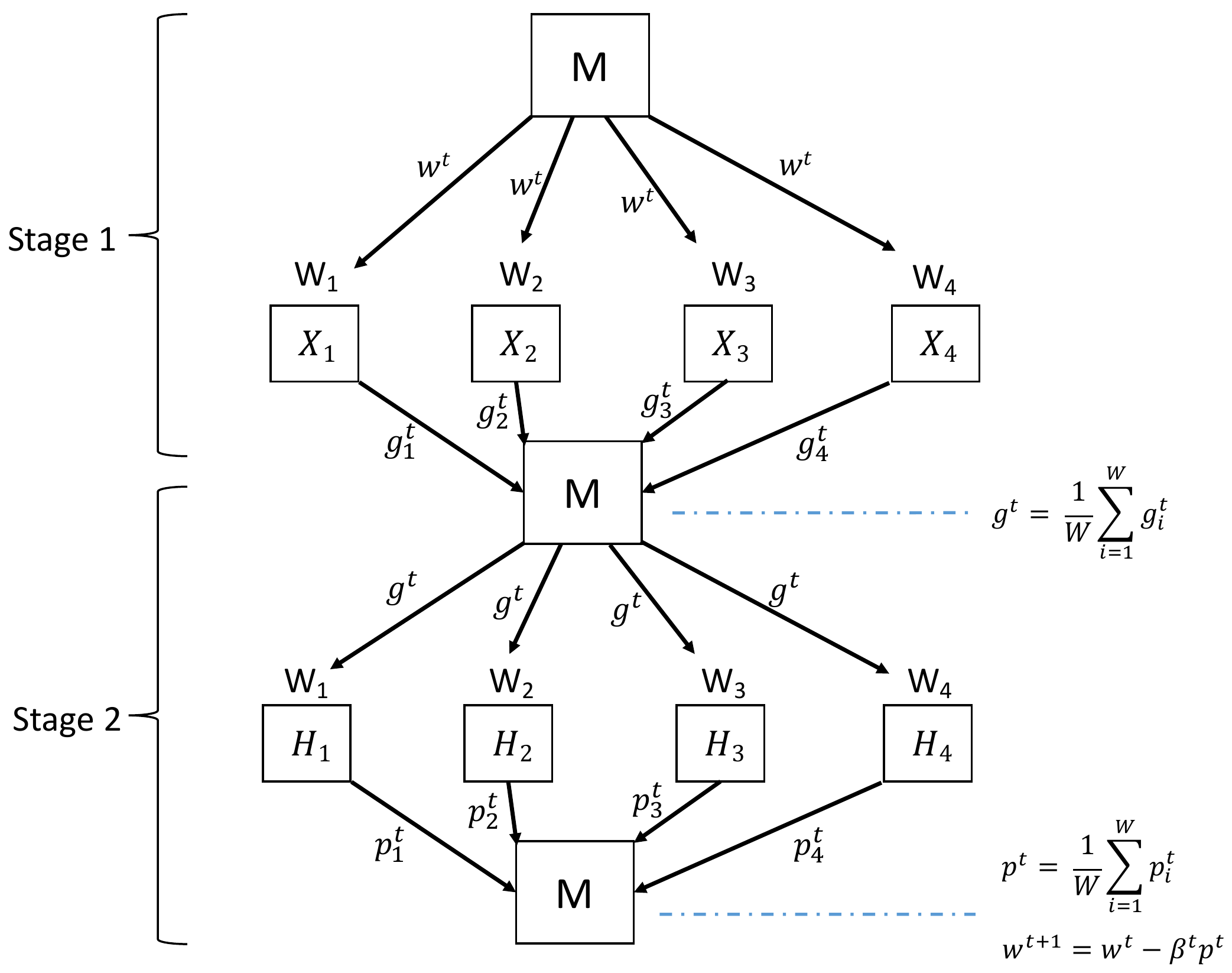}
 \caption{GIANT: The two stage second order distributed optimization scheme with four workers. First, master calculates the full gradient by aggregating local gradients from workers. Second, the master calculates approximate Hessian using local second-order updates from workers.}
 \label{fig:giant}
 \end{figure}

\begin{figure}[t!]
    \centering
    \begin{subfigure}[t]{0.30\textwidth}
        \centering
        \includegraphics[scale=0.14]{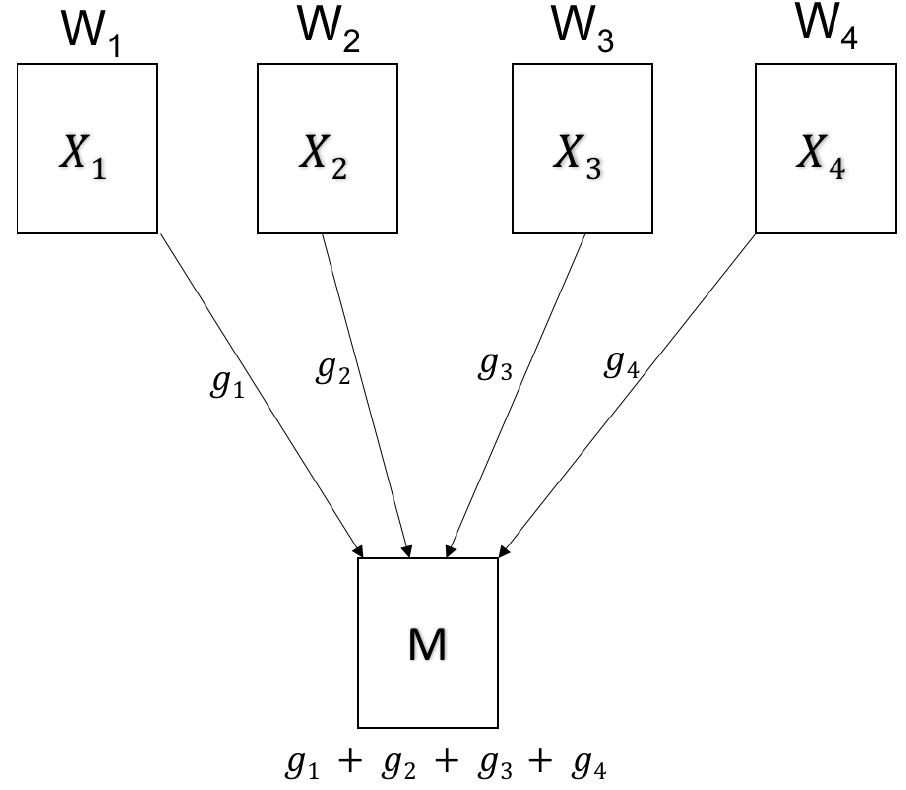}
        \caption{\small Simple Gradient Descent where each worker stores one-fourth fraction of the whole data and sends back a partial gradient corresponding to its own data to the master}
        \label{fig:grad_descent}
    \end{subfigure}%
    ~~ 
    \begin{subfigure}[t]{0.30\textwidth}
        \centering
        \includegraphics[scale=0.14]{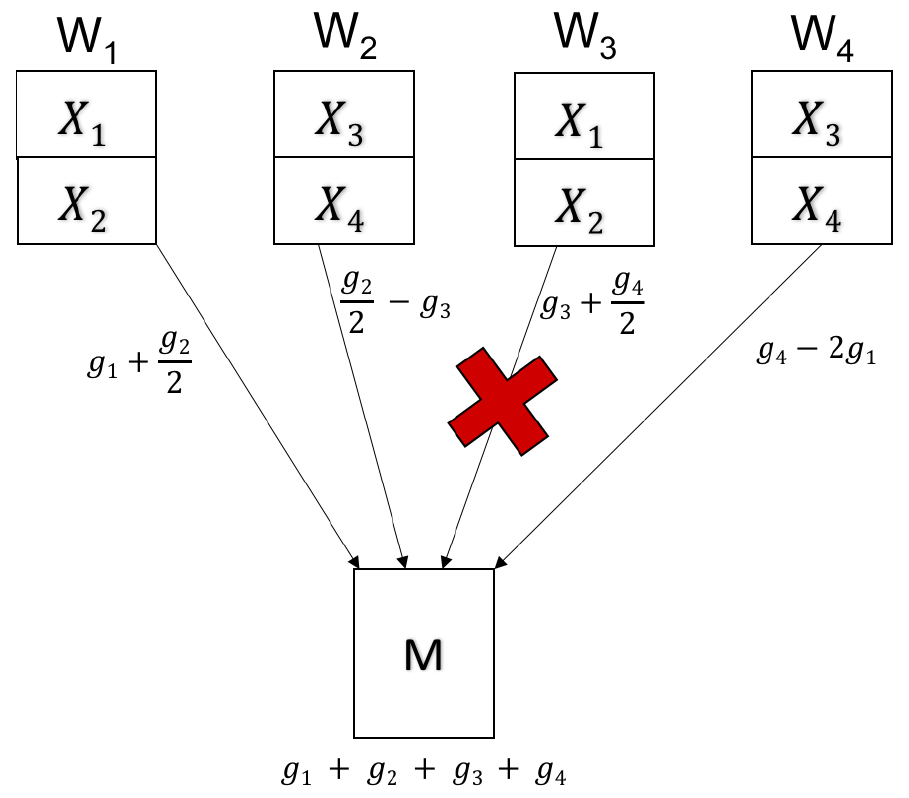}
        \caption{\small Gradient Coding described in \cite{grad_coding} with $W_3$ straggling. To get the global gradient, master would compute $g_1 + g_2 + g_3 + g_4 = 3\left(g_1 + \frac{g_2}{2}\right) - \left(\frac{g2}{2} - g_3\right) + (g_4 - 2g_1)$}
        \label{fig:grad_coding}
    \end{subfigure}
    ~~
    \begin{subfigure}[t]{0.30\textwidth}
        \centering
        \includegraphics[scale=0.31]{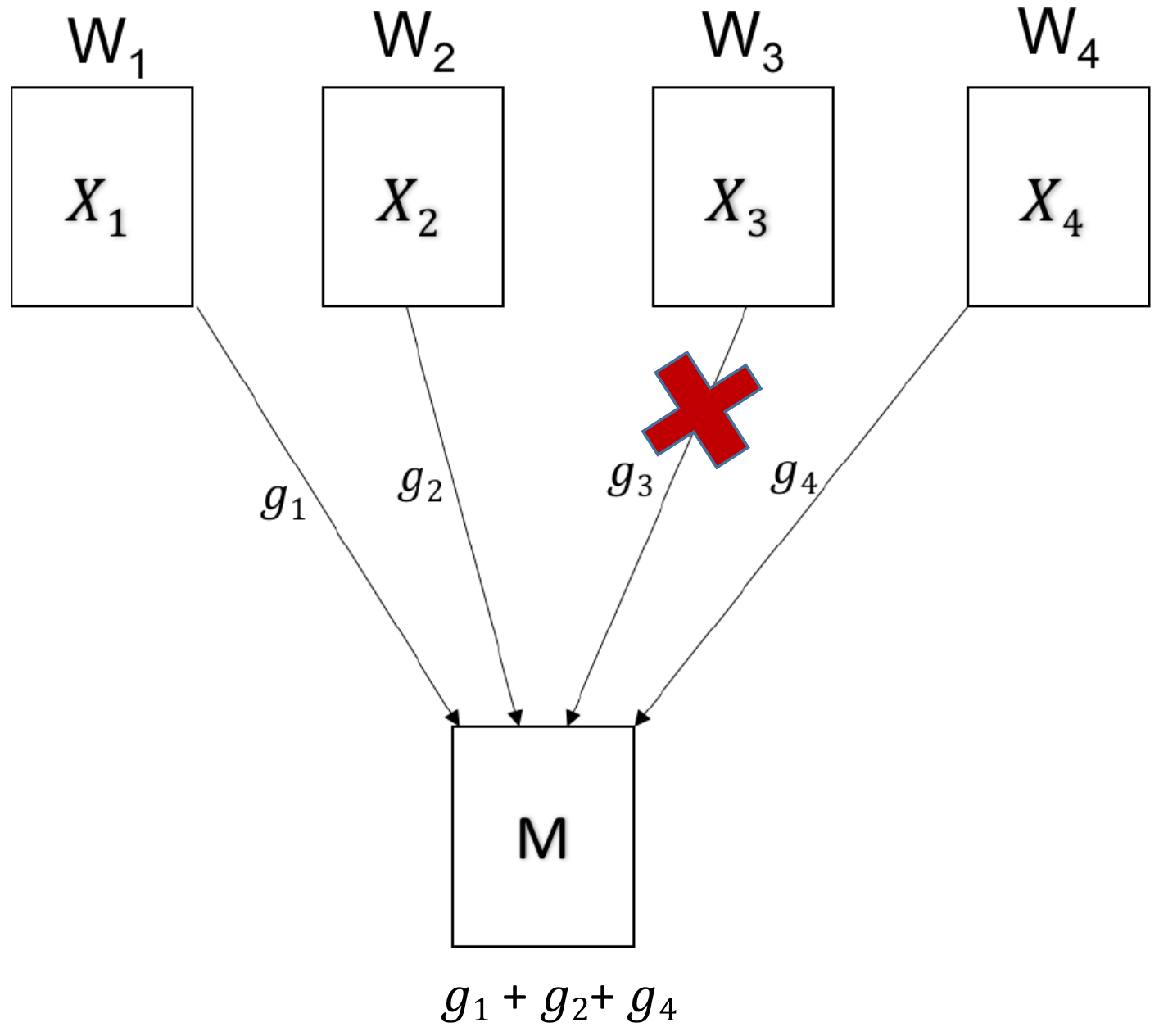}
        \caption{\small Mini-batch gradient descent, where the stragglers are ignored during gradient aggregation and the gradient is later scaled according to the size of mini-batch}
        \label{fig:stoc_grad_coding}
    \end{subfigure}
    \caption{Different gradient descent schemes in serverful systems in presence of stragglers}
    \label{fig:grad_descent_all}
\end{figure}

For comparison of OverSketched Newton with existing distributed optimization schemes, we choose recently-proposed Globally Improved Approximate Newton Direction (GIANT) \cite{giant}. 
The reason is that GIANT boasts a better convergence rate than many existing distributed second-order methods for linear and logistic regression, when $n\gg d$. 
In GIANT, and other similar distributed second-order algorithms, the training data is evenly divided among workers, and the algorithms proceed in two stages. First, the workers compute partial gradients using local training data, which is then aggregated by the master to compute the exact gradient. Second, the workers receive the full gradient to calculate their local second-order estimate, which is then averaged by the master. 
An illustration is shown in Fig. \ref{fig:giant}. 

\begin{figure}
\centering
        \includegraphics[width=0.46\textwidth]{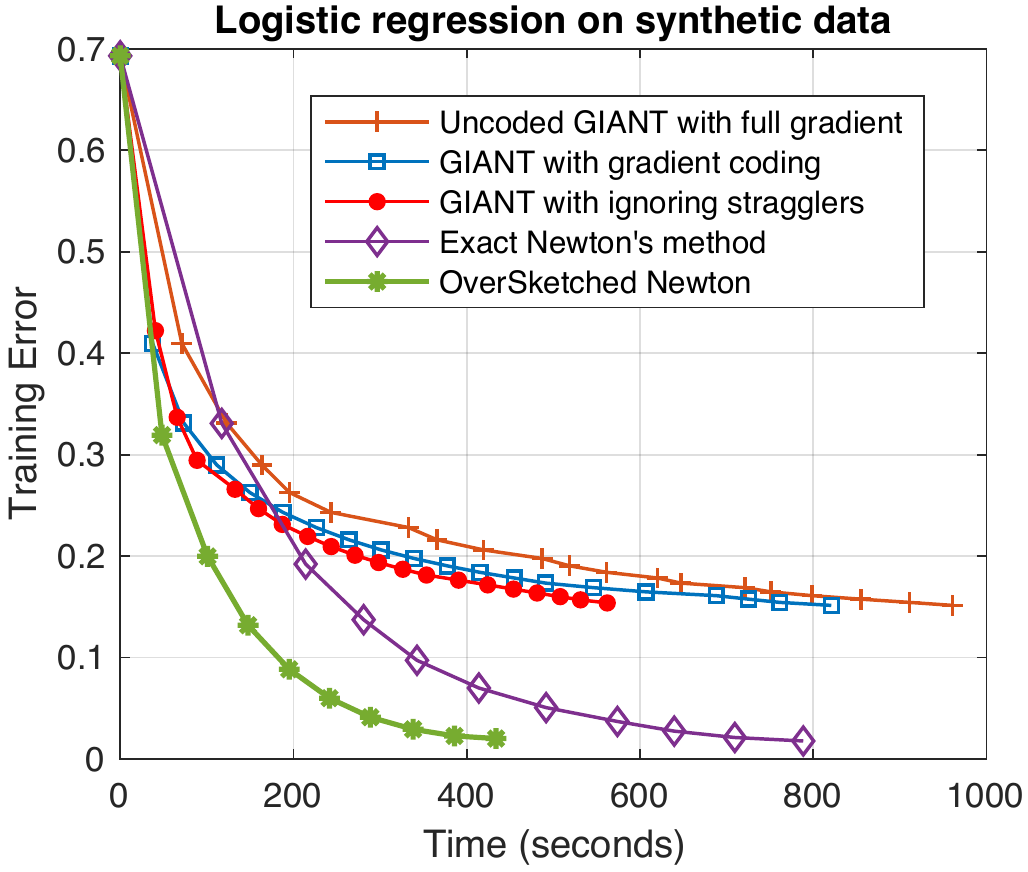}
    \caption{Convergence comparison of GIANT (employed with different straggler mitigation methods), exact Newton's method and OverSketched Newton for Logistic regression on AWS Lambda. The synthetic dataset considered has 300,000 examples and 3000 features.}
\label{fig:second_order_synthetic}
\end{figure}

For straggler mitigation in such serverful systems based algorithms, \cite{grad_coding} proposes a scheme for coding gradient updates called \emph{gradient coding}, where the data at each worker is repeated multiple times to compute redundant copies of the gradient. See Figure \ref{fig:grad_coding} for illustration. Figure \ref{fig:grad_descent} illustrates the scheme that waits for all workers and Figure \ref{fig:stoc_grad_coding} illustrates the ignoring stragglers approach. We use the three schemes for dealing with stragglers illustrated in Figure \ref{fig:grad_descent_all} during the two stages of GIANT, and we compare their convergence with OverSketched Newton.
We further evaluate and compare the convergence exact Newton's method (employed with speculative execution, that is, reassigning and recomputing the work for straggling workers).

\subsection{Comparisons with Existing Second-Order Methods on AWS Lambda}

In Figure \ref{fig:second_order_synthetic}, we present our results on a synthetic dataset with $n = 300,000$ and $d=3000$ for logistic regression on AWS Lambda. 
Each column $\x_i \in \R^d$, for all $i\in [1,n]$, is sampled uniformly randomly from the cube $[-1,1]^d$. 
The labels $y_i$ are sampled from the logistic model, that is, $\mathbb{P}[y_i=1] = 1/(1 + exp(\x_i\w + b))$, where the weight vector $\w$ and bias $b$ are generated randomly from the normal distribution.  ht vector $\w$ and bias $b$ are generated randomly from the normal distribution.    

The orange, blue and red curves demonstrate the convergence for GIANT with the full gradient (that waits for all the workers), {gradient coding} and mini-batch gradient (that ignores the stragglers while calculating gradient and second-order updates) schemes, respectively. The purple and green curves depict the convergence for the exact Newton's method and OverSketched Newton, respectively. The gradient coding scheme is applied for one straggler, that is the data is repeated twice at each worker. We use 60 Lambda workers for executing GIANT in parallel. Similarly, for Newton's method, we use 60 workers for matrix-vector multiplication in steps 4 and 8 of Algorithm \ref{algo:logistic_lambda}, $3600$ workers for exact Hessian computation and $600$ workers for sketched Hessian computation with a sketch dimension of $10d = 30,000$ in step 14 of Algorithm \ref{algo:logistic_lambda}. In all cases, unit step-size was used to update the model\footnote{Line-search in Section \ref{sec:osn} was mainly introduced to prove theoretical guarantees. In our experiments, we observe that constant step-size works well for OverSketched Newton.}

\begin{remark}
\label{rem:no-of-workers}
In our experiments, we choose the number of workers in such a way that each worker receives approximately the same amount of data to work with, regardless of the algorithm. This is motivated by the fact that the memory at each worker is the bottleneck in serverless systems (e.g., in AWS Lambda, the memory at each worker can be as low as 128 MB). Note that this is unlike serverful/HPC systems, where the number of workers is the bottleneck. 
\end{remark}

An important point to note from Fig. \ref{fig:second_order_synthetic} is that the uncoded scheme (that is, the one that waits for all stragglers) has the worst performance. 
 The implication is that good straggler/fault mitigation algorithms are essential for computing in the serverless setting. Secondly, the mini-batch scheme outperforms the gradient coding scheme by $25\%$. This is because gradient coding requires additional communication of data to serverless workers (twice when coding for one straggler, see \cite{grad_coding} for details) at each invocation to AWS Lambda. 
On the other hand, the exact Newton's method converges much faster than GIANT, even though it requires more time per iteration. 

The number of iterations needed for convergence for OverSketched Newton and exact Newton (that exactly computes the Hessian) is similar, but OverSketched Newton converges in almost half the time due to an efficient computation of (approximate) Hessian (which is the computational bottleneck and thus reduces time per iteration). 

\begin{figure*}[t]
    \centering
    \begin{subfigure}[t]{0.47\textwidth}
        \centering
        \includegraphics[scale=0.5]{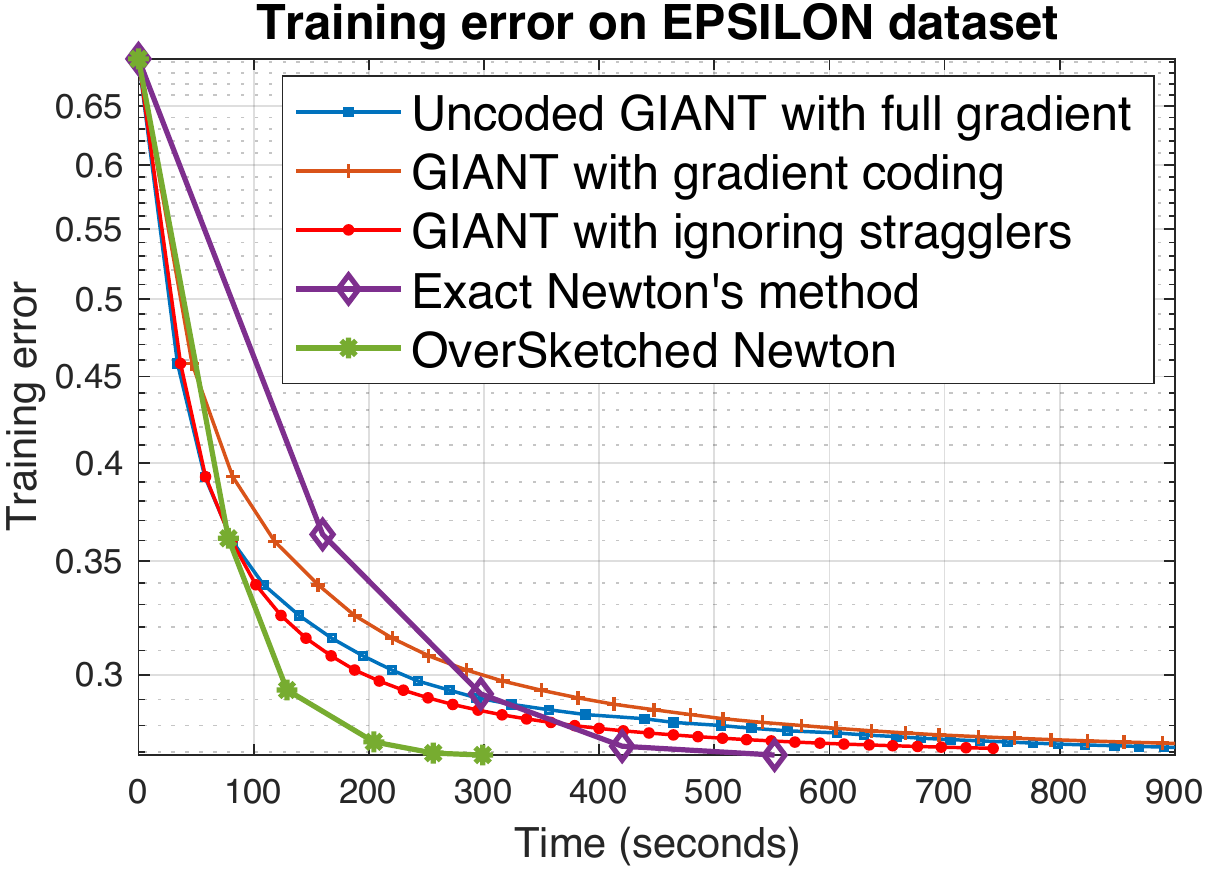}
        \caption{Training error for logistic regression on EPSILON dataset}
    \end{subfigure}
    ~~~~~
    \begin{subfigure}[t]{0.47\textwidth}
        \centering
        \includegraphics[scale=0.5]{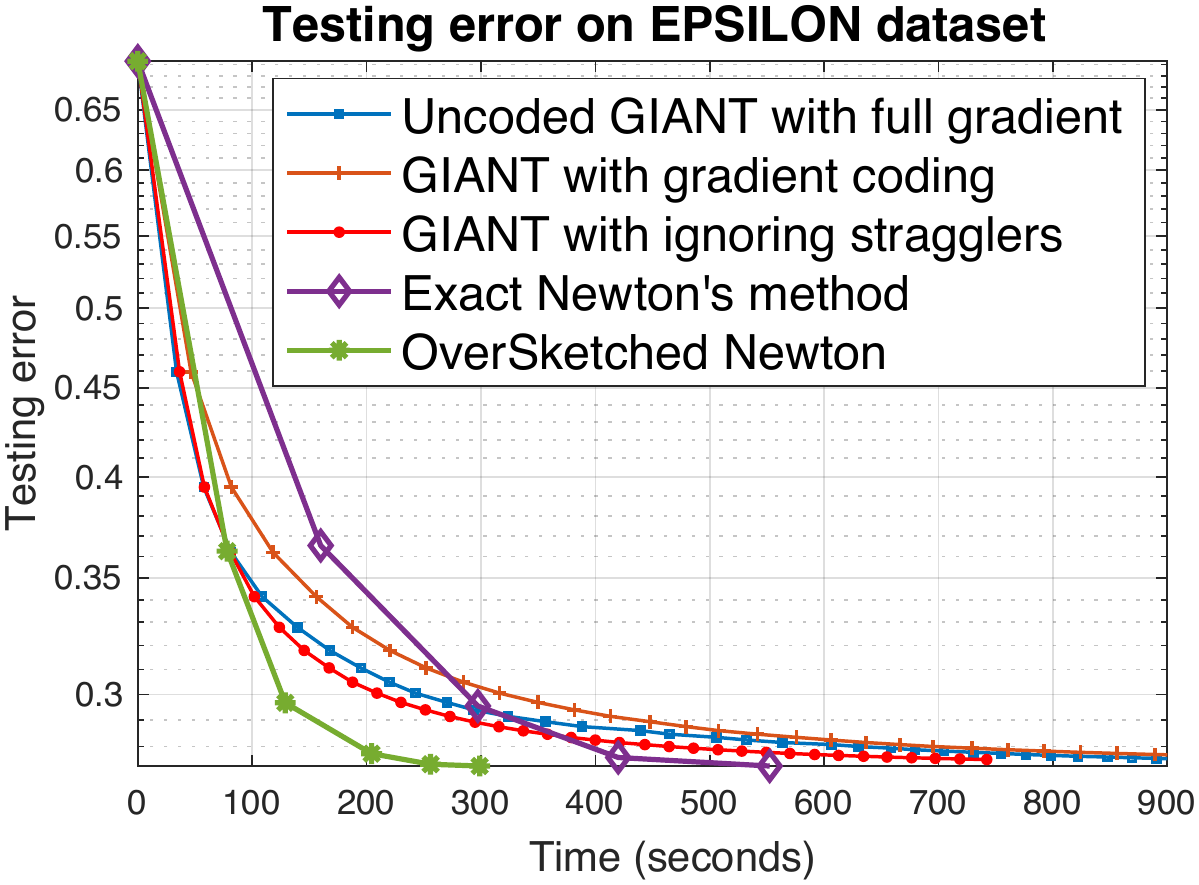}
        \caption{Testing error for logistic regression on EPSILON dataset}
    \end{subfigure}
    \caption{{Comparison of training and testing errors for logistic regression on EPSILON dataset with several Newton based schemes on AWS Lambda. OverSketched Newton outperforms others by at least $46\%$. Testing error closely follows training error.}}
\label{fig:second_order_epsilon}
\end{figure*} 


\begin{figure*}[t]
    \centering
    \begin{subfigure}[t]{0.47\textwidth}
        \centering
        \includegraphics[scale=0.5]{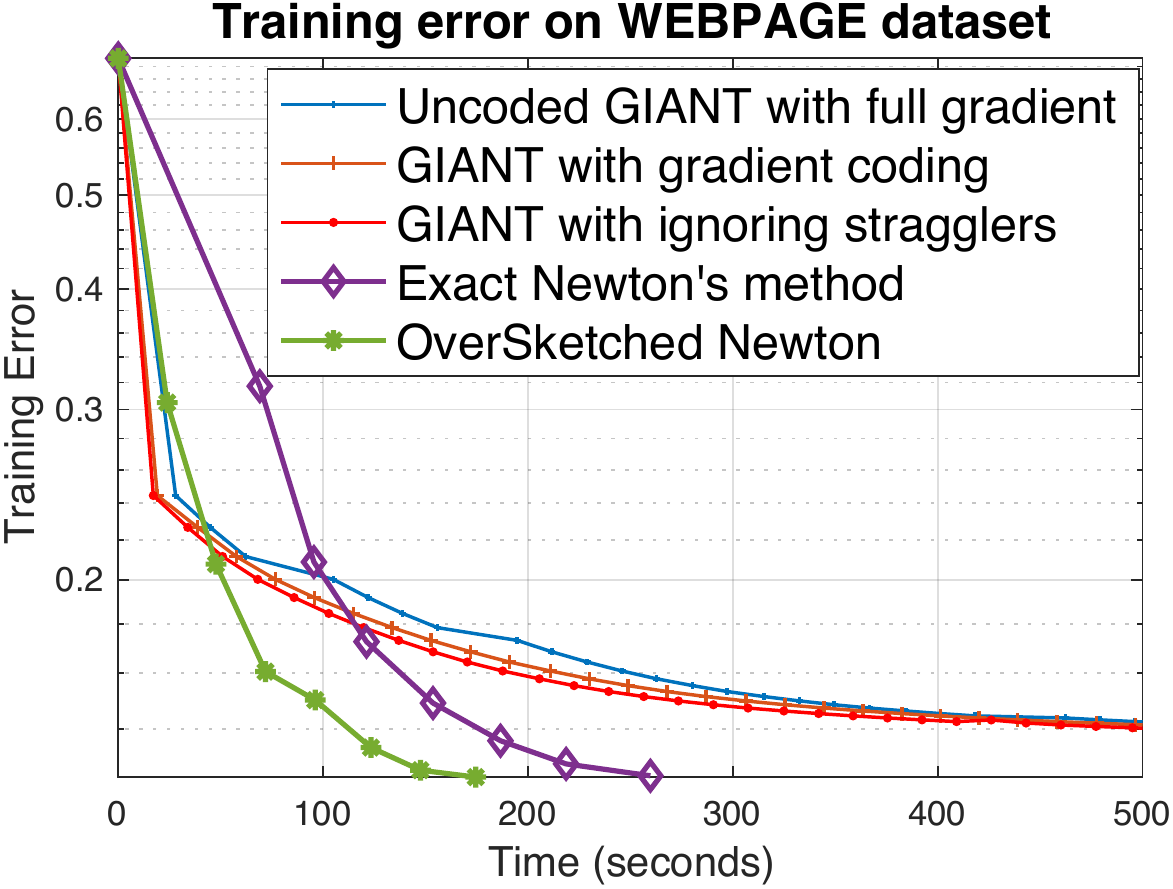}
        \caption{Logistic regression on WEBPAGE dataset}
    \end{subfigure}
    ~~~~~
    \begin{subfigure}[t]{0.47\textwidth}
        \centering
        \includegraphics[scale=0.65]{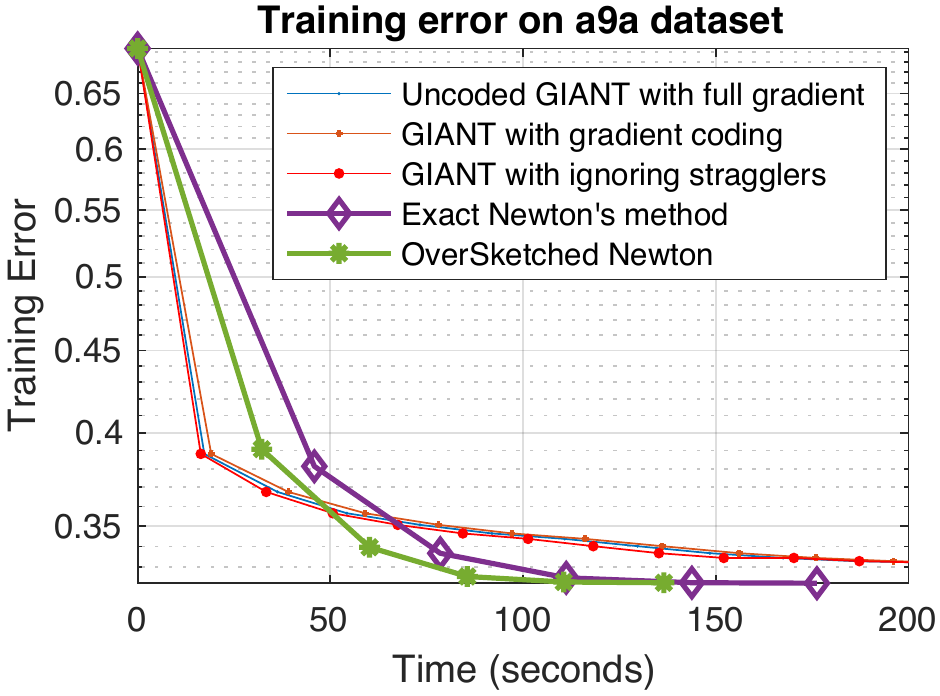}
        \caption{Logistic regression on a9a dataset}
    \end{subfigure}
    \caption{{Logistic regression on WEBPAGE and a9a datasets with several Newton based schemes on AWS Lambda. OverSketched Newton outperforms others by at least $25\%$.}}
\label{fig:second_order_web}
\end{figure*} 

\subsubsection{Logistic Regression on EPSILON, WEBPAGE and a9a Datasets}
In Figure \ref{fig:second_order_epsilon}, we repeat the above experiment with EPSILON classification dataset obtained from \cite{libsvm}, with $n=0.4~million$ and $d=2000$. We plot training and testing errors for logistic regression for the schemes described in the previous section. 
Here, we use $100$ workers for GIANT, and $100$ workers for matrix-vector multiplications for gradient calculation in OverSketched Newton. 
We use gradient coding designed for three stragglers in GIANT. This scheme performs worse than uncoded GIANT that waits for all the stragglers due to the repetition of training data at workers. Hence, one can conclude that the communication costs dominate the straggling costs. In fact, it can be observed that the mini-batch gradient scheme that ignores the stragglers outperforms the gradient coding and uncoded schemes for GIANT. 

During exact Hessian computation, we use $10,000$ serverless workers with speculative execution to mitigate stragglers (i.e., recomputing the straggling jobs) compared to
OverSketched Newton that uses $1500$ workers with a sketch dimension of $15d = 30,000$. 
OverSketched Newton requires a significantly smaller number of workers, as once the square root of Hessian is sketched in a distributed fashion, it can be copied into local memory of the master due to dimension reduction, and the Hessian can be calculated locally.
Testing error follows training error closely, and important conclusions remain the same as in Figure \ref{fig:second_order_synthetic}. OverSketched Newton outperforms GIANT and exact Newton-based optimization by at least $46\%$ in terms of running time.

We repeated the above experiments for classification on the WEBPAGE ($n=49,749$ and $d=300$) and a9a ($n=32,561$ and $d=123$) datasets \cite{libsvm}. For both datasets, we used 30 workers for each iteration in GIANT and any matrix-vector multiplications. Exact hessian calculation invokes $900$ workers as opposed to $300$ workers for OverSketched Newton, where the sketch dimension was $10d = 3000$. 
The results for training loss on logistic regression are shown in Figure  \ref{fig:second_order_web}. Testing error closely follows the training error in both cases. OverSketched Newton outperforms exact Newton and GIANT by at least $\sim 25\%$ and $\sim 75\%$, respectively, which is similar to the trends witnessed heretofore.

\begin{remark}
Note that conventional distributed second-order methods for serverful systems---which distribute training examples evenly across workers (such as~\cite{dane, aide, disco, giant,jaggi18,cocoa})---typically find a ``localized approximation'' (localized to each machine) of second-order update at each worker and then aggregate it. OverSketched Newton, on the other hand, uses the massive storage and compute power in serverless systems to find a more ``globalized approximation'' (globalized in the sense of across machine).  Thus, it performs better in practice.
\end{remark}

\begin{figure}[t!]

\begin{minipage}{0.48\textwidth}
\centering
        \includegraphics[scale=0.5]{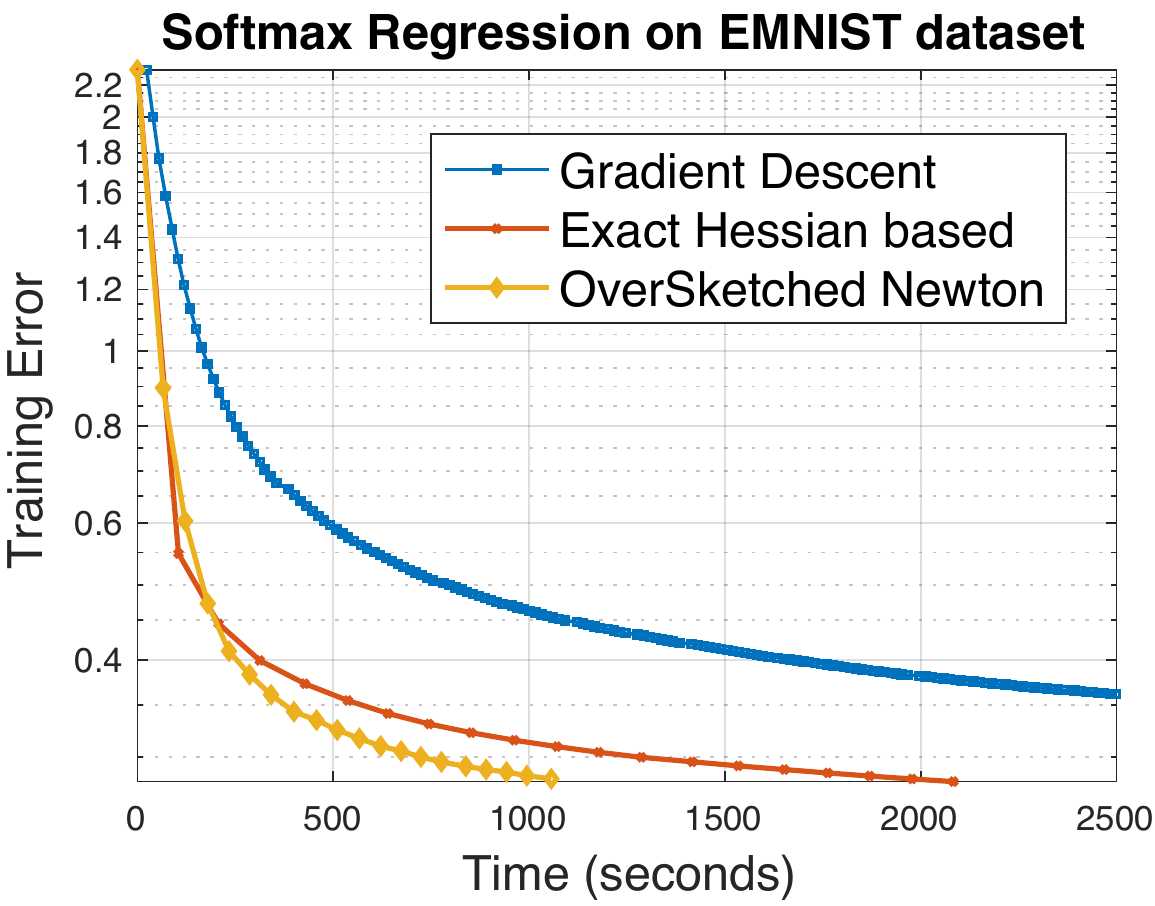}
    \caption{{Convergence comparison of gradient descent, exact Newton's method and OverSketched Newton for Softmax regression on AWS Lambda.}}
\label{fig:softmax}
\end{minipage}
~
\begin{minipage}{0.48\textwidth}
        \centering
        \includegraphics[scale=0.5]{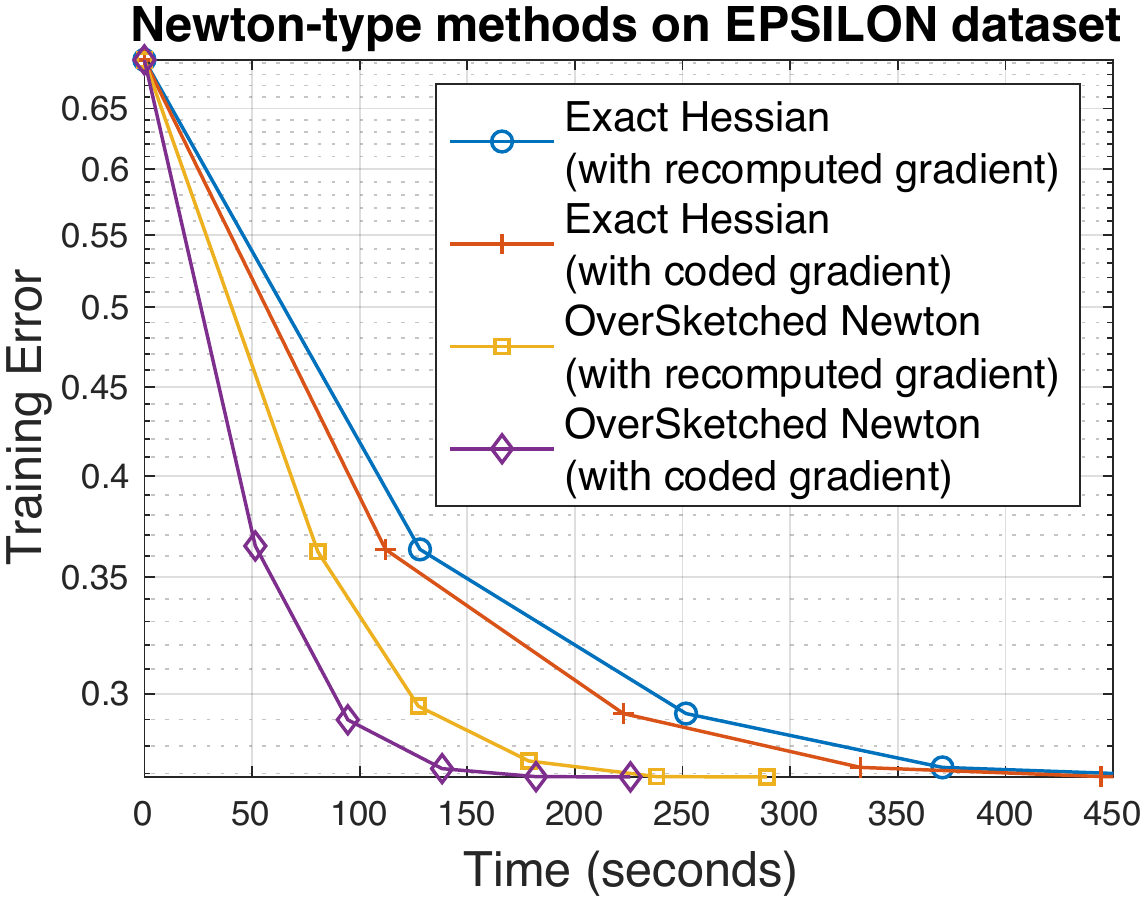}
        \caption{{Convergence comparison of speculative execution and coded computing for gradient and Hessian computing with logistic regression on AWS Lambda.
        }}
    \label{fig:coded_vs_recompute}
    \end{minipage}

    \begin{minipage}{0.48\textwidth}
        \centering
        \includegraphics[scale=0.5]{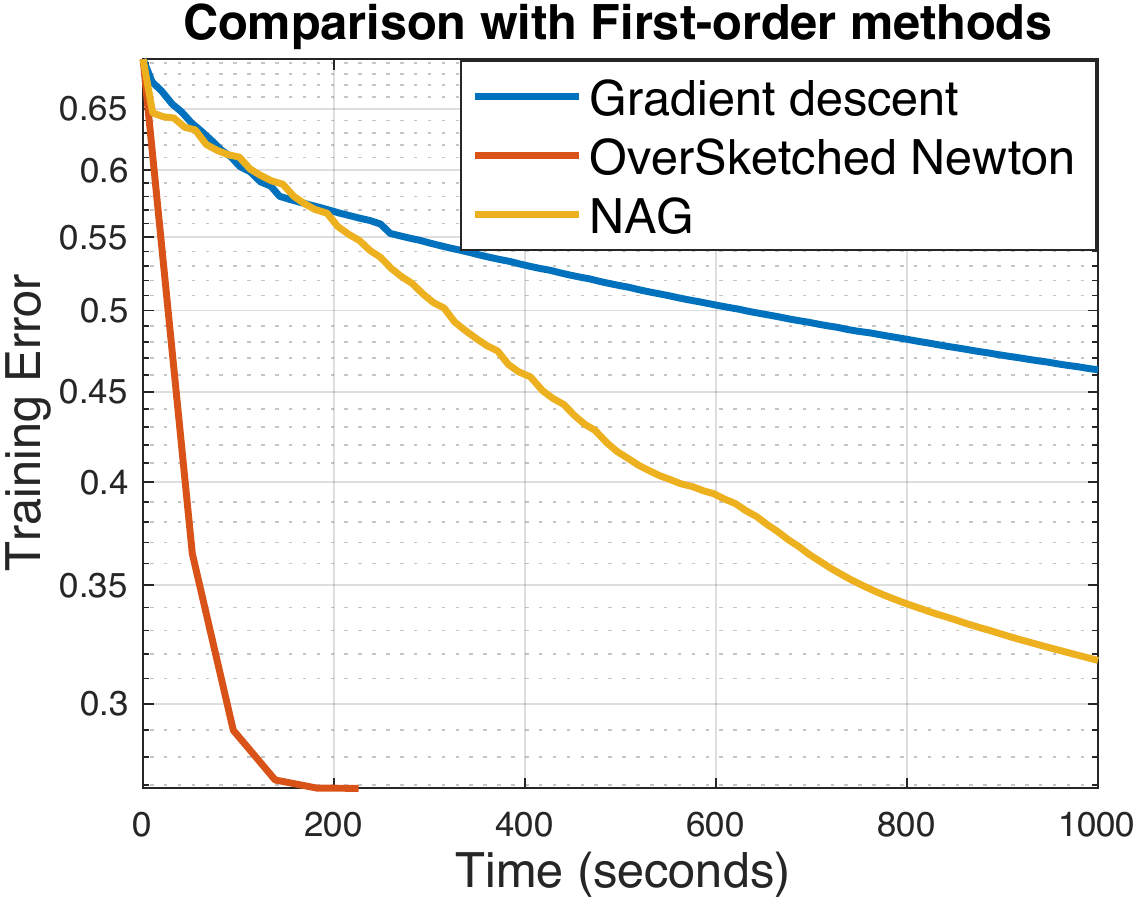}
    \caption{{Convergence comparison of gradient descent, NAG and OverSketched Newton on AWS Lambda.}}
\label{fig:grad_vs_cgosh}
\end{minipage}
~
\begin{minipage}{0.48\textwidth}
    \centering
        \includegraphics[scale=0.53]{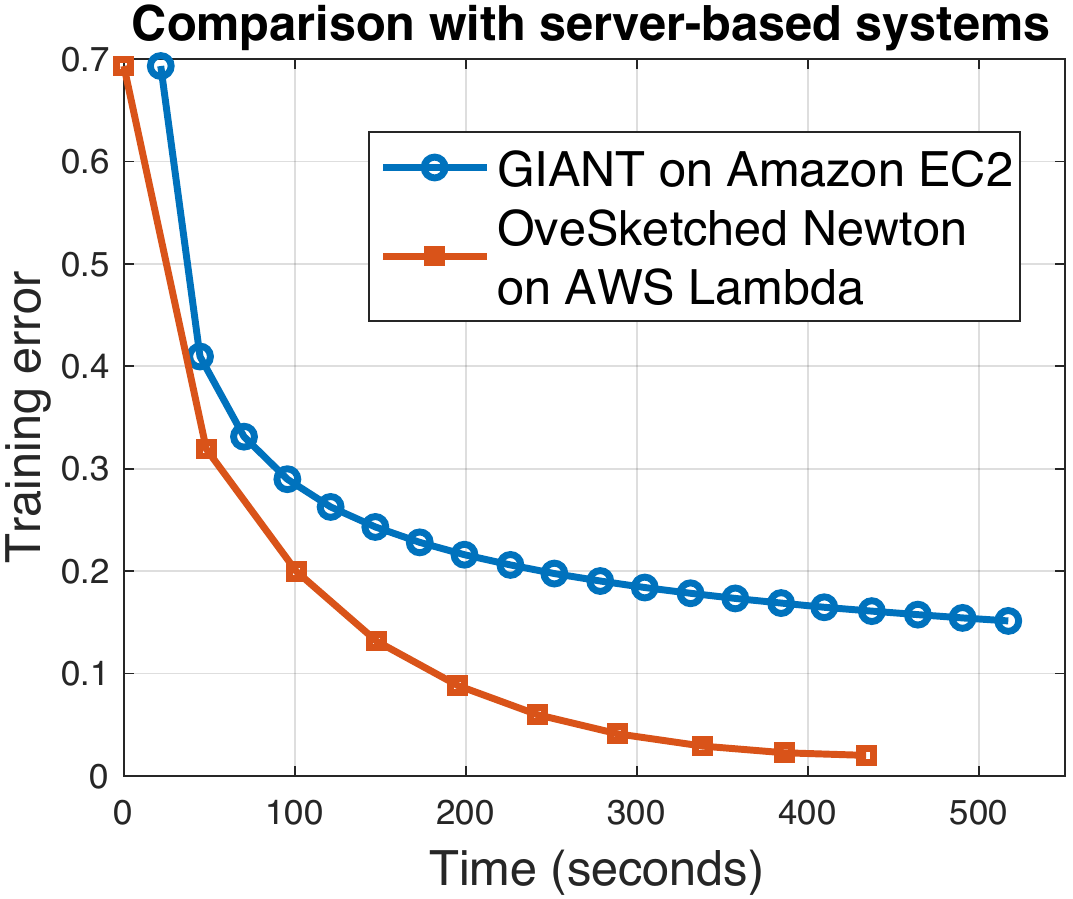}
    \caption{{ Convergence comparison of GIANT on AWS EC2 and OverSketched Newton on AWS Lambda.}}
\label{fig:mpi_vs_lambda}
\end{minipage}
\end{figure} 

\subsection{Softmax Regression on EMIST}



In Fig. \ref{fig:softmax}, we solve unregularized softmax regression, which is weakly convex (see Sec. \ref{sec:softmax} for details). We use the Extended MNIST (EMNIST) dataset \cite{emnist} with $N = 240,000$ training examples, $d = 784$ features and $K=10$ classes. Note that GIANT cannot be applied here as the objective function is not strongly convex. We compare the convergence rate of OverSketched Newton, exact Hessian and gradient descent based schemes. 

For gradient computation in all three schemes, we use 60 workers. However, exact Newton scheme requires 3600 workers to calculate the $dK\times dK$ Hessian and recomputes the straggling jobs, while OverSketched Newton requires only 360 workers to calculate the sketch in parallel with sketch dimension $6dK = 47,040$.
The approximate Hessian is then computed locally at the master using its sketched square root, where the sketch dimension is $6dK = 47,040$. 
The step-size is fixed and is determined by hyperparamter tuning before the start of the algorithm.
Even for the weakly-convex case, second-order methods tend to perform better.
Moreover, the runtime of OverSketched Newton outperforms both gradient descent and Exact Newton based methods by $\sim75\%$ and $\sim50\%$, respectively. 


\subsection{Coded computing versus Speculative Execution}


In Figure \ref{fig:coded_vs_recompute}, we compare the effect of straggler mitigation schemes, namely speculative execution, that is, restarting the jobs with straggling workers, and coded computing on the convergence rate during training and testing. 
We regard OverSketch based matrix multiplication as a coding scheme in which some redundancy is introduced during ``over'' sketching for matrix multiplication. There are four different cases, corresponding to gradient and hessian calculation using either speculative execution or coded computing.
For speculative execution, we wait for at least $90\%$ of the workers to return (this works well as the number of stragglers is generally less than $10\%$) and restart the jobs that did not return till this point. 

For both exact Hessian and OverSketched Newton, using codes for distributed gradient computation outperforms speculative execution based straggler mitigation. 
Moreover, computing the Hessian using OverSketch is significantly better than exact computation in terms of running time as calculating the Hessian is the computational bottleneck in each iteration. 

\subsection{Comparison with First-Order Methods on AWS Lambda}

In Figure \ref{fig:grad_vs_cgosh}, we compare gradient descent and Nesterov Accelerated Gradient (NAG) (while ignoring the stragglers) with OverSketched Newton for logistic regression on EPSILON dataset. 
We observed that for first-order methods, there is only a slight difference in convergence for a mini-batch gradient when the batch size is $95\%$. 
Hence, for gradient descent and NAG, we use 100 workers in each iteration while ignoring the stragglers.%
\footnote{We note that stochastic methods such as SGD perform worse that gradient descent since their update quality is poor, requiring more iterations (hence, more communication) to converge while not using the massive compute power of serverless. For example, 20\% minibatch SGD in the setup of Fig. \ref{fig:grad_vs_cgosh} requires $1.9\times$ more time than gradient descent with same number of workers.} 
These first-order methods were given the additional advantage of backtracking line-search, which determined the optimal amount to move in given a descent direction.%
\footnote{We remark that backtracking line-search required $\sim 13\%$ of the total time for NAG. Hence, as can be seen from Fig. \ref{fig:grad_vs_cgosh}, any well-tuned step-size method would still be significantly slower than OverSketched Newton.} 
Overall, OverSketched Newton with unit step-size significantly outperforms gradient descent and NAG with backtracking line-search.


\subsection{Comparison with Serverful Optimization}

In Fig. \ref{fig:mpi_vs_lambda}, we compare OverSketched Newton on AWS Lambda with existing distributed optimization algorithm GIANT in serverful systems (AWS EC2). 
The results are plotted on synthetically generated data for logistic regression.
For serverful programming, we use Message Passing Interface (MPI) with one \texttt{c3.8xlarge} master and $60$ \texttt{t2.medium} workers in AWS EC2. In \cite{numpywren}, the authors observed that many large-scale linear algebra operations on serverless systems take at least $30\%$ more time compared to MPI-based computation on serverful systems. However, as shown in Fig. \ref{fig:mpi_vs_lambda}, we observe a slightly surprising trend that OverSketched Newton outperforms MPI-based optimization (that uses existing state-of-the-art optimization algorithm). This is because OverSketched Newton exploits the flexibility and massive scale at disposal in serverless, and thus produces a better approximation of the second-order update than GIANT.%
\footnote{We do not compare with exact Newton in serverful sytems since the data is large and stored in the cloud. Computing the exact Hessian would require a large number of workers (e.g., we use 10,000 workers for exact Newton in EPSILON dataset) which is infeasible in existing serverful systems.}
\section{Proofs}


To complete the proofs in this section, we will need the following lemma. 

\begin{lemma}\label{lemma1}
Let $\hh_t = \A_t^T\s_t\s_t^T\A_t$ where $\s_t$ is the sparse sketch matrix in \eqref{sketch_matrix} with sketch dimension $m = \Omega(d^{1+\mu}/\epsilon^2)$ and $N = \Theta_\mu(1/\epsilon)$. Then, the following holds 
\begin{align}
\lambda_{\min}(\hh_t) \geq (1-\epsilon) \lambda_{\min}(\nabla^2 f(\w_t)), \label{lambda_min}\\
\lambda_{\max}(\hh_t) \leq (1+\epsilon)\lambda_{\max}(\nabla^2 f(\w_t)) \label{lambda_max}
\end{align}
with probability at least $1-\frac{1}{d^\tau}$, where $\tau > 0$ is a constant depending on $\mu$ and the constants in $\Theta(\cdot)$ and $\Omega(\cdot)$, and $\lambda_{\max}(\cdot)$ and $\lambda_{\min}(\cdot)$ denote the maximum and minimum eigenvalues, respectively.
In general, 
\begin{align*}
\lambda_i(\nabla^2 f(\w_t)) -\epsilon\lambda_{\max}(\nabla^2 f(\w_t)) \leq \lambda_i(\hh_t) \leq \lambda_i(\nabla^2 f(\w_t)) +\epsilon\lambda_{\max}(\nabla^2 f(\w_t)),
\end{align*}
where $\lambda_i(\cdot)$ is the $i$-th eigenvalue.
\end{lemma}

\begin{proof}
We note than $N$ is the number of non-zero elements per row in the sketch $\s_t$ in \eqref{sketch_matrix} after ignoring stragglers.
We use Theorem 8 in \cite{nelson} to bound the singular values for the sparse sketch $\s_t$ in \eqref{sketch_matrix} with sketch dimension $m =  \Omega(d^{1+\mu}/\epsilon^2)$ and $N = \Theta(1/\epsilon)$. It says that  $\mathbb P (\forall ~\x ~\in~ \R^{n}, ||\s_t\x||_2 \in (1\pm \epsilon/3)||\x||_2) > 1 - 1/d^\tau$, where $\tau > 0$ depends on  $\mu$ and the constants in $\Theta(\cdot)$ and $\Omega(\cdot)$. Thus, $||\s_t\x||_2 \in (1\pm \epsilon/3)||\x||_2$, which implies that
\begin{align*}
 ||\s_t\x||_2^2 \in (1+ \epsilon^2/9 \pm 2\epsilon/3)||\x||_2^2,
\end{align*}
with probability at least $1 - 1/d^\tau$.
For $\epsilon \leq 1/2$, this leads to the following inequality
\begin{align}
||\s_t\x||_2^2 \in (1\pm \epsilon)||\x||_2^2 \Rightarrow |\x^T(\s_t\s_t^T - \I)\x| \leq \epsilon||\x||_2^2 ~\forall ~x\in \R^n \label{lambda_sst}
\end{align}
with probability at least $1 - 1/d^\tau$.
Also, since $(1-\epsilon)\x^T\x \leq \x^T\s_t\s_t^T\x \leq (1+\epsilon)\x^T\x ~\forall ~x\in \R^n$ by the inequality above, replacing $\x$ by $\A\y$, we get
\begin{align}\label{ineq}
(1-\epsilon)\y^T\A^T\A\y \leq \y^T\A^T\s_t\s_t^T\A\y \leq (1+\epsilon)\y^T\A^T\A\y 
\end{align}
with probability at least $1 - 1/d^\tau$. Let $\y_1$ be the unit norm eigenvector corresponding to the minimum eigenvalue of $\hh_t = \A^T_t\s_t\s_t^T\A_t$. Since the above inequality is true for all $\y$, we have
\begin{align*}
\y_1^T\A_t^T\s_t\s_t^T\A_t\y_1 \geq (1-\epsilon)\y_1^T\A_t^T\A_t\y_1 &\geq (1-\epsilon) \lambda_{\min} (\A^T_t\A_t) = (1-\epsilon)\lambda_{\min}(\nabla^2 f(\w_t)) \\
\Rightarrow \lambda_{\min}(\hh_t) &\geq (1-\epsilon)\lambda_{\min}(\nabla^2 f(\w_t))
\end{align*}
with probability at least $1 - 1/d^\tau$. Along similar lines, we can prove that $\lambda_{\max}(\hh_t) \leq (1-\epsilon)\lambda_{\max}(\nabla^2 f(\w_t))$ with probability at least $1 - 1/d^\tau$ using the right hand inequality in \eqref{ineq}. Together, these prove the first result. 

In general, Eq. \eqref{lambda_sst} implies that the eigenvalues of $(\s_t\s_t^T - \I)$ are in the set $[-\epsilon,\epsilon]$. Thus, all the eigenvalues of $\A^T_t(\s_t\s_t^T - \I)\A_t$ are in the set $[-\epsilon\lambda_{\max}(\nabla^2 f(\w_t)),\epsilon\lambda_{\max}(\nabla^2 f(\w_t))]$ Also, we can write 
\begin{align*}
 \hh_t = \A_t^T\s_t\s_t^T\A_t  = \A_t^T\A_t + \A_t^T(\s_t\s_t^T - \I)\A_t.
\end{align*}
Now, applying Weyl's inequality (see \cite{tao_weyl}, Section 1.3) on symmetric matrices $\hh_t = \A_t^T\s_t\s_t^T\A_t$, $\nabla^2 f(\w_t) = \A_t^T\A_t$ and $\A_t^T(\s_t\s_t^T - \I)\A_t$, we get
\begin{align*}
\lambda_i(\nabla^2 f(\w_t)) -\epsilon\lambda_{\max}(\nabla^2 f(\w_t)) \leq \lambda_i(\hh_t) \leq \lambda_i(\nabla^2 f(\w_t)) +\epsilon\lambda_{\max}(\nabla^2 f(\w_t)),
\end{align*}
which proves the second result.

\end{proof}



\subsection{Proof of Theorem \ref{global_thm}}\label{proof_global_conv}

Let's define $\w_{\tau} = \w_t + \tau\p_t$, where the descent direction $\p_t$ is given by $\p_t = -\hh_t^{-1}\nabla f(\w_t)$. Also, from Lemma \ref{lemma1}, we have 
$$\lambda_{\min}(\hh_t)\geq (1-\epsilon)\lambda_{\min}(\nabla^2f(\w_t)) \text{ and } \lambda_{\max}(\hh_t)\leq (1+\epsilon)\lambda_{\max}(\nabla^2f(\w_t)),$$
with probability at least $1 - 1/d^\tau.$ 
Using the above inequalities and the fact that $f(\cdot)$ is $k$-strongly convex and $M$-smooth, we get
\begin{align}\label{eq:eigs_hh}
(1-\epsilon)k\I\preceq \hh_t\preceq (1+\epsilon)M\I,
\end{align}
with probability at least $1 - 1/d^\tau.$ 

Next, we show that there exists an $\alpha>0$ such that the Armijo line search condition in \eqref{eq:line-search} is satisfied.  From the smoothness of $f(\cdot)$, we get (see \cite{nesterov_book}, Theorem 2.1.5) 
\begin{align*}
f(\w_{\alpha}) - f(\w_t) &\leq (\w_\alpha - \w_t)^T\nabla f(\w_t) + \frac{M}{2}||\w_\alpha - \w_t||^2,\\
&=\alpha\p_t^T\nabla f(\w_t) + \alpha^2\frac{M}{2}||\p_t||^2.
\end{align*}
Now, for $\w_{\alpha}$ to satisfy the Armijo rule, $\alpha$ should satisfy
\begin{align*}
\alpha\p_t^T\nabla f(\w_t) + \alpha^2\frac{M}{2}||\p_t||^2 &\leq \alpha\beta\p_t^T\nabla f(\w_t)\\
\Rightarrow \alpha\frac{M}{2}||\p_t||^2 &\leq (\beta - 1)\p_t^T\nabla f(\w_t)\\
\Rightarrow \alpha\frac{M}{2}||\p_t||^2 &\leq (1-\beta)\p_t^T\hh_t\p_t,
\end{align*} 
where the last inequality follows from the definition of $\p_t$.
Now, using the lower bound from \eqref{eq:eigs_hh}, $\w_\alpha$ satisfies Armijo rule for all
$$\alpha \leq \frac{2(1-\beta)(1-\epsilon)k}{M}.$$
Hence, we can always find an $\alpha_t\geq \frac{2(1-\beta)(1-\epsilon)k}{M}$ using backtracking line search such that $\w_{t+1}$ satisfies the Armijo condition, that is,
\begin{align}
f(\w_{t+1}) - f(\w_t) &\leq \alpha_t\beta\p_t^T\nabla f(\w_t)\nn\\
&= -\alpha_t\beta\nabla f(\w_t)^T\hh_t^{-1} \nabla f(\w_t)\nn\\
&\leq -\frac{\alpha_t\beta}{\lambda_{\max} (\hh_t)}||\nabla f(\w_t)||^2\nn
\end{align}
which in turn implies
\begin{align}\label{eq:armijo}
f(\w_{t}) - f(\w_{t+1}) \geq \frac{\alpha_t\beta}{M(1+\epsilon)}||\nabla f(\w_t)||^2
\end{align}
with probability at least $1 - 1/d^\tau$. Here the last inequality follows from the bound in \eqref{eq:eigs_hh}. Moreover, $k$-strong convexity of $f(\cdot)$ implies (see \cite{nesterov_book}, Theorem 2.1.10)
\begin{align*}
f(\w_t) - f(\w^*)\leq \frac{1}{2k}||\nabla f(\w_t)||^2.
\end{align*}
Using the inequality from \eqref{eq:armijo} in the above inequality, we get
\begin{align*}
f(\w_{t}) - f(\w_{t+1}) &\geq \frac{2\alpha_t\beta k}{M(1+ \epsilon)} (f(\w_t) - f(\w^*))\\
&\geq \rho (f(\w_t) - f(\w^*)),
\end{align*}
where $\rho = \frac{2\alpha_t\beta k}{M(1+ \epsilon)}$. Rearranging, we get 
\begin{align*}
f(\w_{t+1}) - f(\w^*) \leq (1-\rho)(f(\w_t) - f(\w^*))
\end{align*}
with probability at least $1 - 1/d^\tau$, which proves the desired result.

\subsection{Proof of Theorem \ref{convergence_thm}}\label{proof_local_conv}

According to OverSketched Newton update, $\w_{t+1}$ is obtained by solving
\begin{align*}
\w_{t+1} = \arg\min_{\w\in \R^d} \Big\{f(\w_t) + \nabla f(\w_t)^T(\w - \w_t) 
+ \frac{1}{2}(\w - \w_t)^T\hh_t(\w - \w_t)\Big\}.
\end{align*}
Thus, we have, for any $\w\in \R^d$,
\begin{align*}
& f(\w_t) + \nabla f(\w_t)^T(\w - \w_t) + \frac{1}{2}(\w - \w_t)^T\hh_t(\w - \w_t), \\
&~~~~~~~\geq f(\w_t) + \nabla f(\w_t)^T(\w_{t+1} - \w_t) + \frac{1}{2}(\w_{t+1} - \w_t)^T\hh(\w_{t+1} - \w_t), \\
&\Rightarrow \nabla f(\w_t)^T(\w - \w_{t+1}) + \frac{1}{2}(\w - \w_t)^T\hh_t(\w - \w_t) - \frac{1}{2}(\w_{t+1} - \w_t)^T\hh_t(\w_{t+1} - \w_t)  \geq 0,\\
&\Rightarrow \nabla f(\w_t)^T(\w - \w_{t+1}) + \frac{1}{2}\Big[(\w - \w_t)^T\hh_t(\w - \w_{t+1}) + (\w - \w_{t+1})^T\hh_t(\w_{t+1} - \w_t)\Big]  \geq 0.
\end{align*}
Substituting $\w$ by $\w^*$ in the above expression and calling $\Delta_t = \w^* - \w_{t}$, we get
\begin{align*}
& -\nabla f(\w_t)^T\Delta_{t+1} + \frac{1}{2}\Big[\Delta_{t+1}^T\hh_t(2\Delta_t - \Delta_{t+1})\Big]  \geq 0,\\
&\Rightarrow \Delta_{t+1}^T\hh_t\Delta_t -\nabla f(\w_t)^T\Delta_{t+1} \geq \frac{1}{2}\Delta_{t+1}^T\hh_t \Delta_{t+1}.
\end{align*}
Now, due to the optimality of $\w^*$, we have $\nabla f(\w^*)^T\Delta_{t+1}\geq 0$. Hence, we can write
\begin{align*}
&\Delta_{t+1}^T\hh_t\Delta_t -(\nabla f(\w_t) - \nabla f(\w^*))^T\Delta_{t+1} \geq \frac{1}{2}\Delta_{t+1}^T\hh_t \Delta_{t+1}.
\end{align*}
Next, substituting $\nabla f(\w_t) - \nabla f(\w^*) = \Big( \int_0^1\nabla^2 f(\w^* + p(\w_t - \w^*))dp\Big)(\w_t - \w^*)$ in the above inequality, we get 
\begin{align*}
\Delta_{t+1}^T(\hh_t -  \nabla^2 f(\w_t)) \Delta_t + \Delta_{t+1}^T\Big(\nabla^2 f(\w_t) - \int_0^1\nabla^2 f(\w^* + p(\w_t - \w^*))dp\Big)\Delta_t  \geq \frac{1}{2}\Delta_{t+1}^T\hh_t \Delta_{t+1}.
\end{align*}
Using Cauchy-Schwartz inequality in the LHS above, we get
\begin{align*}
||\Delta_{t+1}||_2 ||\Delta_t||_2 \Big(||\hh_t -  \nabla^2 f(\w_t)||_{2}+  \int_0^1||\nabla^2 f(\w_t) -\nabla^2 f(\w^* + p(\w_t - \w^*))||_{2}dp\Big)  \geq \frac{1}{2}\Delta_{t+1}^T\hh_t \Delta_{t+1}.
\end{align*}
Now, using the $L$-Lipschitzness of $\nabla^2f(\cdot)$ in the inequality above, we get
\begin{align}\label{ineq11}
\frac{1}{2}\Delta_{t+1}^T\hh_t \Delta_{t+1} &\leq ||\Delta_{t+1}||_2 ||\Delta_t||_2||\hh_t -  \nabla^2 f(\w_t)||_{2} + \frac{L}{2}||\Delta_{t+1}||_2 ||\Delta_t||_2^2 \int_0^1(1-p)dp, \nn\\
\Rightarrow \frac{1}{2}\Delta_{t+1}^T\hh_t \Delta_{t+1} &\leq ||\Delta_{t+1}||_2\Big(||\Delta_t||_2||\hh_t -  \nabla^2 f(\w_t)||_{2} + \frac{L}{2} ||\Delta_t||_2^2\Big).
\end{align}

Note that for the positive definite matrix $\nabla^2 f(\w_t) = \A_t^T\A_t$, we have $||\A_t||^2_{2} = ||\nabla^2 f(\w_t)||_{2}$. Moreover,
\begin{align*}
||\hh_t -  \nabla^2 f(\w_t)||_{2} = ||\A_t^T(\s_t\s_t^T - \I)\A_t||_{2} \leq ||\A_t||^2_{2}||\s_t\s_t^T - \I||_{2}
\end{align*} 
Now, using Equation \ref{lambda_sst} from the proof of Lemma \ref{lemma1}, we get $||\s_t\s_t^T - \I||_{2} = \lambda_{\max}(\s_t\s_t^T - \I) \leq \epsilon$. Using this to bound the RHS of \eqref{ineq11}, we have, with probability at least $1 - 1/d^\tau$,
\begin{align*}
 \frac{1}{2}\Delta_{t+1}^T\hh_t \Delta_{t+1} &\leq ||\Delta_{t+1}||_2\Big( \epsilon||\nabla^2 f(\w_t)||_{2}||\Delta_t||_2+ \frac{L}{2} ||\Delta_t||_2^2\Big)\\
 \frac{1}{2} ||\s_t\A\Delta_{t+1}||_2^2 &\leq ||\Delta_{t+1}||_2\Big(\epsilon||\nabla^2 f(\w_t)||_{2}||\Delta_t||_2+ \frac{L}{2} ||\Delta_t||_2^2\Big),
 \end{align*} 
where the last inequality follows from $\hh_t = \A_t^T\s_t^T\s_t\A_t$. Now, since the sketch dimension $m = \Omega(d^{1+\mu}/\epsilon^2)$,  
using Eq. \eqref{lambda_sst} from the proof of Lemma 1 in above inequality, we get, with probability at least $1 - 1/d^\tau$,
\begin{align*}
 \frac{1}{2}(1-\epsilon) ||\A\Delta_{t+1}||_2^2 &\leq ||\Delta_{t+1}||_2\Big(\epsilon||\nabla^2 f(\w_t)||_{2}||\Delta_t||_2+ \frac{L}{2} ||\Delta_t||_2^2\Big),\\
 \Rightarrow \:\frac{1}{2}(1-\epsilon) \Delta_{t+1}^T\nabla^2 f(\w_t)\Delta_{t+1} &\leq ||\Delta_{t+1}||_2\Big(\epsilon||\nabla^2 f(\w_t)||_{2}||\Delta_t||_2+ \frac{L}{2} ||\Delta_t||_2^2\Big).
 \end{align*}

\noindent Now, since $\gamma$ and $\beta$ are the minimum and maximum eigenvalues of  $\nabla^2 f(\w^*)$, we get
 \begin{align*}
 \frac{1}{2}(1-\epsilon) ||\Delta_{t+1}||_2(\gamma - L||\Delta_{t}||_2) &\leq\epsilon(\beta + L||\Delta_t||_2)||\Delta_t||_2+ \frac{L}{2} ||\Delta_t||_2^2
 \end{align*} 
 by the Lipschitzness of $\nabla^2 f(\w)$, that is, $|\Delta_{t+1}^T(\nabla^2 f(\w_t) - \nabla^2 f(\w^*))\Delta_{t+1}| \leq L||\Delta_t||_2||\Delta||_{t+1}^2$. Rearranging for $\epsilon\leq \gamma/(8\beta) < 1/2$, we get
 \begin{align}\label{recursion}
 ||\Delta_{t+1}||_2 \leq \frac{4\epsilon\beta}{\gamma - L||\Delta_t||_2}||\Delta_t||_2 + \frac{5L}{2(\gamma - L||\Delta_t||_2)} ||\Delta_t||_2^2,
 \end{align}
 with probability at least $1 - 1/d^\tau$.

 Let $\xi_T$ be the event that the above inequality (in \eqref{recursion}) is true for $t = 0,1,\cdots, T$. 
Thus, $$\mathbb{P}(\xi_T) \geq \Big(1 - \frac{1}{d^\tau}\Big)^T \geq 1 - \frac{T}{d^\tau},$$
where the second inequality follows from Bernoulli's inequality.
 Next, assuming that the event $\xi_T$ holds, we prove that $||\Delta_t||_2 \leq \gamma/5L$ using induction.
 We can verify the base case using the initialization condition, i.e. $||\Delta_0||_2 \leq \gamma/8L$. Now, assuming that $||\Delta_{t-1}||_2 \leq \gamma/5L$ and using it in the inequality \eqref{recursion}, we get
 \begin{align*}
  ||\Delta_{t}||_2 &\leq \frac{4\epsilon\beta}{\gamma}\times\frac{\gamma}{5L} + \frac{5L}{2\gamma}\times\frac{\gamma^2}{25L^2} \\
  &= \frac{4\epsilon\beta}{5L} + \frac{\gamma}{10L}\\
  &\leq \frac{\gamma}{L}\bigg(\frac{1}{10} + \frac{1}{10}\bigg) \leq \frac{\gamma}{5L}, 
 \end{align*}
where the last inequality uses the fact that $\epsilon \leq \gamma/(8\beta)$. Thus, by induction, 
$$ ||\Delta_t||_2 \leq \gamma/(5L) ~~\forall~~ t\geq 0 ~~~~~\text{with probability at least}~ 1 - {T}/{d^\tau}.$$  
Using this in \eqref{recursion}, we get the desired result, that is,
\begin{align*}
 ||\Delta_{t+1}||_2 \leq \frac{5\epsilon\beta}{\gamma}||\Delta_t||_2 + \frac{25L}{8\gamma} ||\Delta_t||_2^2,
\end{align*}
with probability at least $1 - T/d^\tau$.

\subsection{Proof of Theorem \ref{global_thm_weakly_convex}}\label{proof_global_conv_weakly_convex}
Let us define a few short notations for convenience. Say $\g_t = \nabla f(\w_t)$ and $\h_t = \nabla^2f(\w_t) = \A_t^T\A_t$, and we know that $\hh_t = \A^T_t\s_t\s^T_t\A_t$. Moreover, all the results with approximate Hessian $\hh_t$ hold with probability $1 - 1/d^\tau$. We skip its mention in most of the proof for brevity. The following lemmas will assist us in the proof.
\begin{lemma}\label{lemma:smoothness_for_g2}
$M$-smoothness of $f(\cdot)$ and $L$-Lipchitzness of $\nabla^2 f(\cdot)$ imply 
\begin{align}
\left\|\nabla^{2} f(\mathbf{y}) \nabla f(\mathbf{y})-\nabla^{2} f(\mathbf{x}) \nabla f(\mathbf{x})\right\| \leq Q||\y - \x||
\end{align}
for all $\x\in \R^d$, $Q = (L\delta + M^2)$, where 
$\y \in \mathcal{Y}$, where $\mathcal{Y} = \{\y\in \R^d|~||\nabla f(\mathbf{y})|| \leq \delta\}$ and $\delta > 0$ is some constant. 
\end{lemma}
\begin{proof}
We have
\begin{align*}
LHS &= \left\|\nabla^{2} f(\mathbf{y}) \nabla f(\mathbf{y})-\nabla^{2} f(\mathbf{x}) \nabla f(\mathbf{x})\right\| \\
&= \left\|\nabla^{2} f(\y) - \nabla^2 f(\mathbf{x}))\nabla f(\y) + \nabla^2 f(\mathbf{x})(\nabla f(\mathbf{y}) - \nabla f(\mathbf{x})) \right\| 
\end{align*}
By applying triangle inequality and Cauchy-Schwarz to above equation, we get
\begin{align*} 
LHS \leq ||\nabla^{2} f(\y) - \nabla^2 f(\mathbf{x})||_2||\nabla f(\y)|| + ||\nabla^2 f(\mathbf{x})||_2||\nabla f(\mathbf{y}) - \nabla f(\mathbf{x})||
\end{align*}
From the smoothness of $f(\cdot)$, that is, Lipshitzness of gradient, we get $||\nabla^2 f(\mathbf{x})||_2 \leq M ~\forall ~x\in\R^d$. Additionally, using Lipshitzness of Hessian, we get
\begin{align*}
LHS &\leq (L||\nabla f(\y)|| + M^2) ||\y - \x||\\
&\leq (L\delta + M^2) ||\y - \x||
\end{align*}
for $\y\in \mathcal{Y}$. This proves the desired result.
\end{proof}

\begin{lemma}\label{lemma:Utg_lower_bound}
  Let $\A^T = \U\sqrt{\Sigma}\V^T$ and $\A^T\s_t = \hU\sqrt{\hat\Sigma}\hat{\V}^T$ be the truncated Singular Value Decompositions (SVD) of $\A^T$ and $\A^T\s_t$, respectively. Thus, $\h_t = \U\Sigma\U^T$ and $\hh_t = \hat{\U}\hat{\Sigma}\hat{\U}^T$. Then, for all $\g \in \R^d$, we have
  \begin{align}
    ||\hU^T\g||^2 \geq \frac{(1-\epsilon)\eta}{M(1+\epsilon)}||\U^T\g||^2,
  \end{align}
  where $\eta$ is defined in Assumption (5).
\end{lemma}
\begin{proof}
For all $\g\in \R^d$, using the fact that $\A = \bf{V}\sqrt{\Sigma}\U^T$, we get
\begin{align}\label{ineq1}
 ||\A\g||^2 &= (\U^T\g)^T\Sigma(\U^T\g)\nn\\
&\geq \lambda_{\min}(\Sigma)||\U^T\g||^2\nn\\
&\geq \eta||\U^T\g||^2,
\end{align}
where the last inequality uses Assumption (5). In a similar fashion, we can obtain
\begin{align}\label{ineq2}
 ||\s^T_t\A\g||^2 &= (\hU^T\g)^T\hat{\Sigma}(\hU^T\g)\nn\\
&\leq \lambda_{\max}(\hat\Sigma)||\hU^T\g||^2\nn \\
&\leq M(1+\epsilon)||\hU^T\g||^2,
\end{align}
where the last inequality uses $M$-smoothness of $f(\cdot)$ and Lemma \ref{lemma1}. Also, from the subspace embedding property of $\s_t$ (see Lemma \ref{lemma1}), we have
\begin{align*}
  ||\s^T\A\g||^2 \geq (1-\epsilon)||\A\g||^2.
\end{align*}
Now, using the above inequality and Eqs. \eqref{ineq1} and \eqref{ineq2}, we get
\begin{align}
    ||\hU^T\g||^2 \geq \frac{(1-\epsilon)\eta}{M(1+\epsilon)}||\U^T\g||^2,
  \end{align}
which is the desired result.
\end{proof}

Now we are ready to prove Theorem \ref{global_thm_weakly_convex}. Let $\h_t = \U\Sigma\U^T$ and $\hh_t = \hat{\U}\hat{\Sigma}\hat{\U}^T$ be the truncated SVDs of $\h_t$ and $\hh_t$, respectively. Also, let $\alpha_t$ be the step-size obtained using line-search in \eqref{eq:line-search2} in the $t$-th iteration. Thus, Eq. \eqref{eq:line-search2} with the update direction $\p_t = -\ddh_t\g_t$ implies
\begin{align}\label{line-search-bound}
 ||\g_{t+1}||^2 &\leq ||\g_t||^2 - 2\beta\alpha_t\langle \hh_t\g_t, \ddh_t\g_t\rangle \nn\\
&= ||\g_t||^2 - 2\beta\alpha_t||\hU_T^T\g_t||^2,
\end{align}
where the last equality uses the fact that $\ddh_t$ can be expressed as  $\ddh_t = \hat{\U}\hat{\Sigma}^{-1}\hat{\U}^T$. Note that Lemma \ref{lemma:smoothness_for_g2} implies that the function $||\nabla f(\y)||^2/2$ is smooth for all $y \in \mathcal{Y}$, where $\mathcal{Y} = \{\y\in \R^d|~||\nabla f(\mathbf{y})|| \leq \delta\}$. Smoothness in turn implies the following property (see \cite{nesterov_book}, Theorem 2.1.10)
\begin{align}\label{smoothness_final}
  \frac{1}{2}||\nabla f(\mathbf{y})||^2 \leq \frac{1}{2}||\nabla f(\mathbf{x})||^2 + \langle \nabla^2 f(\x) ,\y-\x\rangle + \frac{1}{2}Q||\y -\x||^2 ~\forall ~\x, \y\in \mathcal(Y),
\end{align} 
where $Q = L\delta + M^2$. We take $\delta = ||\nabla f(\w_0)||$ where $\w_0$ is the initial point of our algorithm. Due to line-search condition in \eqref{eq:line-search2}, it holds that $||\nabla f(\w_t)|| \leq ||\nabla f(\w_0)||~\forall~ t>0$. Thus, substituting $\x = \w_t$ and $\y = \w_{t+1} = \w_t + \alpha_t\p_t$, we get
\begin{align}\label{ineq_smooth1}
 \frac{1}{2}||\g_{t+1}||^2 &\leq \frac{1}{2}||\g_t||^2 + \langle \h_t\g_t ,\alpha_t\p_t\rangle + \frac{1}{2}Q\alpha^2||\p_t||^2\nn\\
\Rightarrow ||\g_{t+1}||^2 &\leq ||\g_t||^2 + \langle 2\h_t\g_t ,\alpha_t\p_t\rangle + Q\alpha^2||\p_t||^2,
\end{align} 
where 
$$Q = L||\nabla f(\w_0)|| + M^2.$$
Also, since the minimum non-zero eigenvalue of $\h_t \geq\eta$ from Assumption (5), the minimum non-zero eigenvalue of $\hh_t$ is at least $\eta - \epsilon M$ from Lemma \ref{lemma1}. Thus, 
\begin{equation}\label{upper_bound_hhd}
  |\ddh_t||_2 \leq 1/(\eta - \epsilon M).
\end{equation}
|Moreover, 
\begin{equation}\label{norm_pt}
||\p_t|| = ||-\ddh_t\g_t|| \leq ||\ddh_t||_2||\g_t|| \leq \frac{||\g_t||}{(\eta - \epsilon M)}.
\end{equation}
Using this in \eqref{ineq_smooth1}, we get
\begin{align}\label{ineq_smooth2}
||\g_{t+1}||^2 &\leq ||\g_t||^2 - 2\alpha_t\langle \h_t\g_t ,\ddh_t\g_t\rangle + Q\alpha^2\frac{||\g_t||^2}{(\eta - \epsilon M)^2}.
\end{align}
Now, 
\begin{align*}
  -\langle \h_t\g_t ,\ddh_t\g_t\rangle &= -\langle \hh_t\g_t ,\ddh_t\g_t\rangle + \langle (\hh_t - \h_t)\g_t ,\ddh_t\g_t\rangle\\
  \Rightarrow-\langle \h_t\g_t ,\ddh_t\g_t\rangle &\leq -||\hU_t^T\g_t||^2 + ||\g_t||^2||\hh_t - \h_t||_2 ||\ddh_t||_2,
\end{align*} 
where the last inequality is obtained by applying the triangle inequality and Cauchy-Schwartz inequality. This can be further simplified using Lemma \ref{lemma1} and Eq. \eqref{upper_bound_hhd} as{}
\begin{align*}
-\langle \h_t\g_t ,\ddh_t\g_t\rangle &\leq -||\hU_t^T\g_t||^2 +  \frac{\epsilon M}{(\eta - M\epsilon)}||\g_t||^2
\end{align*}
Using the above in Eq. \eqref{ineq_smooth2}, we get
\begin{align}\label{ineq_smooth3}
||\g_{t+1}||^2 &\leq ||\g_t||^2 + 2\alpha_t(-||\hU_t^T\g_t||^2 +  \frac{\epsilon M}{(\eta - M\epsilon)}||\g_t||^2) + Q\alpha_t^2\frac{||\g_t||^2}{(\eta - \epsilon M)^2}  
\end{align}
Note that the upper bound in Eq. \eqref{ineq_smooth3} always holds. Also, we want the inequality in \eqref{line-search-bound} to hold for some $\alpha_t > 0$. 
Therefore, we want $\alpha_t$ to satisfy the following (and hope that it is always satisfied for some $\alpha_t>0$)
\begin{align}\label{ineq_smooth4}
||\g_t||^2 + 2\alpha_t(-||\hU_t^T\g_t||^2 &+  \frac{\epsilon M}{(\eta - M\epsilon)}||\g_t||^2) + Q\alpha_t^2\frac{||\g_t||^2}{(\eta - \epsilon M)^2} \leq ||\g_t||^2 - 2\beta\alpha_t||\hU_T^T\g_t||^2 \nn\\
\Rightarrow Q\alpha_t^2\frac{||\g_t||^2}{(\eta - \epsilon M)^2} &\leq 2\alpha_t\bigg[(1 - \beta)||\hU_t^T\g_t||^2 - \frac{\epsilon M}{(\eta - M\epsilon)}||\g_t||^2\bigg]\nn\\
\Rightarrow \alpha_t &\leq \frac{2(\eta - \epsilon M)^2}{Q}\bigg[(1-\beta)\frac{||\hU_t^T\g_t||^2}{||\g_t||^2} - \frac{\epsilon M}{(\eta - M\epsilon)}\bigg].
\end{align}
Thus, any $\alpha_t$ satisfying the above inequality would definitely satisfy the line-search termination condition in 

Now, using Lemma \ref{lemma:Utg_lower_bound} and Assumption (6), we have
\begin{equation}\label{bound_uthg}
||\hU_t^T\g_t||^2 \geq \frac{(1-\epsilon)\eta}{M(1+\epsilon)}||\U^T_t\g||^2 \geq \frac{(1-\epsilon)\eta}{M(1+\epsilon)}\nu||\g||^2.
\end{equation}
Using the above in Eq. \eqref{ineq_smooth4} to find an iteration independent bound on $\alpha_t$, we get
\begin{align}\label{ineq_smooth5}
\alpha_t &\leq \frac{2(\eta - \epsilon M)^2}{Q}\bigg[(1-\beta)\nu - \frac{\epsilon M}{(\eta - M\epsilon)}\bigg].
\end{align}
Hence, line-search will always terminate for all $\alpha_t$ that satisfy the above inequality. This can be further simplified by assuming that $\epsilon$ is small enough such that $\epsilon < \eta/2M$. Thus, $\eta - M\epsilon > \eta/2$, and the sufficient condition on $\alpha_t$ in \eqref{ineq_smooth5} becomes
\begin{align}\label{ineq_smooth6}
\alpha_t &\leq \frac{\eta}{2Q}\big[(1-\beta)\nu\eta - 2\epsilon M\big]. 
\end{align}
For a positive $\alpha_t$ to always exist, we require $\epsilon$ to further satisfy
\begin{align}
  \epsilon \leq \frac{(1-\beta)\nu\eta}{2M},
\end{align}
which is tighter than the initial upper bound on $\epsilon$. Now, Eqs. \eqref{line-search-bound} and \eqref{bound_uthg} proves the desired result, that is
\begin{align*}
||\g_{t+1}||^2 \leq ||\g_t||^2 - 2\beta\alpha_t||\hU_T^T\g_t||^2 \leq \bigg(1 - 2\beta\alpha_t\nu\frac{(1-\epsilon)\eta}{M(1+\epsilon)}\bigg)||\g_t||^2.
\end{align*}
Thus, OverSketched Newton for the weakly-convex case enjoys a uniform linear convergence rate of decrease in $||\nabla f(\w)||^2$.


\section{Conclusions}

We proposed OverSketched Newton, a straggler-resilient distributed optimization algorithm for serverless systems. It uses the idea of matrix sketching from RandNLA to find an approximate second-order update in each iteration. We proved that OverSketched Newton has a local linear-quadratic convergence rate for the strongly-convex case, where the dependence on the linear term can be made to diminish by increasing the sketch dimension. Moreover, it has a linear global convergence rate for weakly-convex functions.
By exploiting the massive scalability of serverless systems, OverSketched Newton produces a global approximation of the second-order update. Empirically, this translates into faster convergence than state-of-the-art distributed optimization algorithms on AWS Lambda.

\section*{Acknowledgments}
This work was partially supported by NSF grants  CCF-1748585 and CNS-1748692 to SK, and NSF grants CCF-1704967 and CCF- 0939370 (Center for Science of Information) to TC,  and ARO, DARPA, NSF, and ONR grants to MWM, and  NSF Grant CCF-1703678 to KR. The authors would like to additionally thank Fred-Roosta and Yang Liu for helpful discussions regarding our proof techniques and AWS for providing promotional cloud credits for research. 
 
\bibliographystyle{ieeetr}
 \bibliography{bibli}

\end{document}